\documentclass[journal,comsoc]{IEEEtran}

\usepackage{graphicx}
\usepackage{amssymb}
\usepackage{amsmath}
\usepackage{cite}
\usepackage{subfigure}
\usepackage{mathrsfs}
\usepackage[displaymath,mathlines]{lineno}
\usepackage{color}
\usepackage{tabulary}
\usepackage{multirow}
\usepackage{algpseudocode}
\usepackage{algorithm,algpseudocode}
\usepackage{algorithmicx}
\usepackage{pbox}
\usepackage{multicol}
\usepackage{lipsum}
\usepackage{amsthm}
\usepackage{relsize}
\usepackage{lipsum}
\usepackage{epstopdf}
\usepackage{mathtools}
\usepackage{calc}

\makeatletter
\newcommand{\doublewidetilde}[1]{{%
		\mathpalette\double@widetilde{#1}}}
\newcommand{\double@widetilde}[2]{%
		\sbox\z@{$\m@th#1\widetilde{#2}$}%
		\ht\z@=.5\ht\z@
		\widetilde{\box\z@}}
\makeatother

\newtheorem{definition}{Definition}
\newtheorem{theorem}{Theorem}
\newtheorem{lemma}{Lemma}

\newtheorem{remark}{Remark}
\newtheorem{assumption}{Assumption}

\begin{document}
%
\title{Uplink Power Control in Massive MIMO with Double Scattering Channels}

\author{\normalsize Trinh~Van~Chien, \textit{Member}, \textit{IEEE}, Hien~Quoc~Ngo, \textit{Senior Member}, \textit{IEEE}, Symeon~Chatzinotas, \textit{Senior Member}, \textit{IEEE}, Bj\"{o}rn~Ottersten, \textit{Fellow}, \textit{IEEE}, and M\'erouane Debbah, \textit{Fellow}, \textit{IEEE}
\thanks{This   work   was partially supported by RISOTTI-Reconfigurable  Intelligent  Surface  for  Smart  Cities  under  Project FNR/C20/IS/14773976/RISOTTI, and  in  part by the UK Research and Innovation Future Leaders Fellowships under Grant MR/S017666/1. Parts of this paper were presented at IEEE ICC 2021 \cite{Chien2021ICC}. The  associate  editor  coordinating  the  review  of this paper  and approving  it for publication  was A. Liu. \textit{(Corresponding author: Trinh Van Chien.)}}
\thanks{T. V. Chien, S. Chatzinotas, B. Ottersten are with the Interdisciplinary Centre for Security, Reliability and Trust (SnT), University of Luxembourg, L-1855 Luxembourg, Luxembourg (email: vanchien.trinh@uni.lu, symeon.chatzinotas@uni.lu, and bjorn.ottersten@uni.lu).}
\thanks{
H. Q. Ngo is with the School of Electronics, Electrical Engineering and Computer Science, Queen's University Belfast, Belfast BT7 1NN, United Kingdom (email: hien.ngo@qub.ac.uk).}
\thanks{M. Debbah is  with the
	Technology Innovation Institute, Abu Dhabi, United Arab Emirates.  He is also with CentraleSup\'elec,
	University Paris-Saclay, 91192 Gif-sur-Yvette, France. (email:  merouane.debbah@tii.ae).}
\thanks{This paper was accepted to publish in the IEEE Transactions on Wireless Communications. Color  versions  of  one  or  more  of  the  figures  in  this  article  are  available online at http://ieeexplore.ieee.org.}
}

\maketitle

\begin{abstract}
Massive multiple-input multiple-output (MIMO) is a key technology for improving the spectral and energy efficiency in 5G-and-beyond wireless networks. For a tractable analysis, most of the previous works on Massive MIMO have been focused on the system performance with complex Gaussian channel impulse responses under rich-scattering environments. In contrast, this paper investigates the uplink ergodic spectral efficiency (SE) of each user under the double scattering channel model. We derive a closed-form expression of the uplink ergodic SE by exploiting the maximum ratio (MR) combining technique based on imperfect channel state information. We further study the asymptotic SE behaviors as a function of the number of antennas at each base station (BS) and the number of scatterers available at each radio channel. We then formulate and solve a total energy optimization problem for the uplink data transmission that aims at simultaneously satisfying the required SEs from all the users with limited data power resource. Notably, our proposed algorithms can cope with the congestion issue appearing when at least one user is served by lower SE than requested. Numerical results illustrate the effectiveness of the closed-form ergodic SE over Monte-Carlo simulations. Besides, the system can still provide the required SEs to many users even under congestion.
\end{abstract}


%
\IEEEpeerreviewmaketitle

\section{Introduction}
Wireless communications has sustained an exponential demand growth  in data throughput and reliability over the last decades \cite{tariq2020speculative,nguyen2020performance}. The cellular network topology with the assistance of MIMO technology has been evolved over time to indulge the growing demand. However, mobile traffic will increase as foreseen in a short time with $12.3$ billion wireless access devices by $2022$ \cite{index2019global}. To handle this issue, Massive MIMO, a disruptive technology with commercial deployments started in $2018$ \cite{bjornson2019massive}, not only inherits all the  multiplexing gain and spatial diversity of the conventional MIMO but also offers extra degree-of-freedoms as a consequence of equipping base stations (BSs) with many antennas \cite{Marzetta2010a}. Massive MIMO, therefore, provides unprecedented spectral and energy efficiency gains of modern wireless networks with only utilizing the contemporary time and frequency resources. Each Massive MIMO BS only exploits a low-cost linear processing technique such as maximum ratio (MR) or zero-forcing (ZF) combining to detect the transmit signals and obtain performance closed to the optimum thanks to the benefits of the use of many more antennas than users \cite{Hoydis2013a}. In the uplink transmission, combining vectors for data detection are constructed from channel estimates, and therefore, the overhead is only made practically proportional to the number of users by sending pilot signals in the uplink.

In Massive MIMO, the closed-form expression of the ergodic SE can be obtained in certain scenarios. For rich scattering environments such that propagation channels ideally follow uncorrelated Rayleigh fading, the uplink and downlink SEs were obtained as a function of large-scale fading coefficients when each BS exploits MR or ZF combining as in \cite{Ngo2013a, Chien2016b} and references therein. As such, many impacts such as array gains and channel estimation quality are explicitly observed in those ergodic rates, together with the power scaling laws are achieved. However, practical channels usually involve spatial correlation, which is modeled, for example utilizing correlated Rayleigh fading in the isotropic scattering environment where the gathered energy at an antenna array comes from many directions leading to the full ranks of covariance matrices with an overwhelming probability \cite{Hoydis2013a,Chien2018a,kammoun2014linear}. For rank deficiency occurring in poor scattering conditions, the Kronecker channel model is popularly used to describe the spatial correlations at the transmitter and receiver \cite{kammoun2015generalized,yu2002models}. The authors in \cite{Gesbert2002a} proposed the \textit{double scattering channel} and demonstrated that the channel capacity is also characterized by the structure of scattering in the propagation environment instead of the spatial correlations around the transceiver only.
	
A few works have studied the effects of low-rank channels in Massive MIMO communications. For the keyhole channels (uncorrelated and rank-deficient), the channel hardening and favorable propagation were investigated in \cite{ngo2017no} to impress a significant reduction of the ergodic SE compared with that of uncorrelated Rayleigh fading. An extension of this work to communications scenarios where users having multiple antennas has recently reported in \cite{James2020}. The first work numerically studying the uplink ergodic SE of cellular Massive MIMO systems with the double scattering channels (spatially correlated and rank-deficient) was found in \cite{van2016multi}. For theoretical analysis, the authors in \cite{kammoun2019asymptotic,Ye2020} computed the asymptotic ergodic SE of a single-cell Massive MIMO system with the different linear precoding techniques when the number of BS antennas, scatterers, and users grow large with the same rate. It is worth emphasizing that these works assumed each user utilizing an orthogonal pilot signal in a single-cell system and the formulations are asymptotically established. To the best of our knowledge, no prior work analyzes the performance of cellular Massive MIMO systems with a finite number of BS antennas, users, and scatterers, where the ergodic SE might have different features than the one at an asymptotic regime.

Many resource allocation tasks in Massive MIMO communications can be implemented on the large-scale fading time scale in place of the small-scale fading one by virtue of the channel hardening \cite{massivemimobook}. This makes resource allocation feasible to implement in practice. Various optimization problems with different utility functions have been formulated and solved in the Massive MIMO literature \cite{Chien2017b,Jin2015a,nguyen2018optimal}. Notice that the key component of Massive MIMO communications is that it can allow many users to access and share the radio resource at the same time with high quality of service. The max-min fairness optimization is therefore promising to provide uniform service to all the users in the coverage area \cite{Sun2015}.  However, for large-scale networks with many base stations and users, the fairness level will approach a zero rate \cite{ghazanfari2019fair}. In contrast, one can include separate SE constraints in the optimization problems to simultaneously maintain the quality of service for all the users \cite{van2020joint,ngo2018total}. However, since the users were randomly distributed, many user locations with poor channel conditions leads the optimization problems to be infeasible. The preliminary work in \cite{van2020uplink} has indicated that many users are still served by the required SEs if we can detect and relax the constraints of unsatisfied users when solving the problems and analyzing uncorrelated Rayleigh fading only.

By exploiting the double scattering channel model, this paper considers a Massive MIMO system in which a set of orthogonal pilot signals are reused by all the users such that the BSs can estimate channels in the pilot training phase. We then compute the uplink ergodic SE of each user in relation to  the channel structure and propagation environment. The ergodic rate is then used to formulate and solve the total energy optimization problem for the uplink data transmission when each BS uses MR combining to detect the desired signals. Our main contributions are summarized as follows:
\begin{itemize}
	\item A new ergodic SE expression is derived in closed form for a finite number of antennas at each BS while the number of scatterers observed by each user and BS is different from each other. This closed-form SE expression explicitly demonstrates the influence of pilot contamination, channel estimation errors, and limited scatterers. Conforming with the literature, we also analyze the asymptotic closed-form SE expression when the number of antennas and/or scatterers grows large. We analytically testify the existence of a saturated point in most of the scenarios, but although the system still can offer an unbounded capacity under a certain condition.
	\item We formulate a total uplink data energy minimization problem subject to the required SE from every user and the power constraints. This problem may have an infeasible domain under the complication of simultaneously serving many users. For user locations and shadow fading realizations, where our optimization problem is feasible, the global optimum can be obtained in polynomial time owning to its convexity. 
	\item We propose two low computational complexity iterative algorithms that tackle the infeasible optimization problem by relaxing the SE constraints of unsatisfied users. At each iteration, the first algorithm allows users to transmit full data power whenever the required SE constraints are not satisfied. In contrast, the second algorithm gives a procedure to scale down data power assigned to users with the lower SEs than requested.  
	\item Numerical results manifest that the closed-form SE expression overlaps Monte-Carlo simulations in all the system parameter settings. The effectiveness of the proposed data power control algorithms are compared with the interior-point methods. For given user locations and shadow fading realizations that form infeasible problems, the system still can provide satisfactory service to many users after relaxing one or a few the required SE constraints. 
\end{itemize}
This paper is organized as follows: Section~\ref{Sec:SystemModel} presents the considered Massive MIMO system under the double scattering channels and derives the closed-form expression of the uplink SE for the case each BS utilizing MR combining to decode the transmitted signals. We also compute the asymptotic SE as different factors grow large. Section~\ref{Sec:Optimization} formulates the total data energy minimization problem and characterizes its canonical form and feasible domain. The two algorithms to obtain a solution to this problem and handle the congestion issue are proposed in Section~\ref{Sec:Solutions}. Finally, Section~\ref{Sec:NumericalResults} shows extensive numerical results and the main conclusions are drawn in Section~\ref{Sec:Conclusion}.

\textit{Notation}: Upper-case bold face letters are used to denote matrices
and lower-case bold face ones for vectors. $\mathbf{I}_{M}$ is the identity matrix of size $M \times M$. The operation $\mathbb{E} \{ \cdot \}$ and $\mathrm{Var} \{ \cdot \}$ denotes the expectation and variance of a random variable, respectively. The notation $\| \cdot \|$ is the Euclidean norm of a vector and $\| \cdot \|_2$ is the spectral norm of a matrix. Moreover, $\mathrm{tr}(\cdot)$ is the trace of a matrix. The regular and Hermitian transposes
are denoted by $(\cdot)^T$ and $(\cdot)^H$, respectively. Finally, $\mathcal{CN}(\cdot, \cdot)$ denotes the circularly symmetric complex Gaussian distribution.
	
\section{Massive MIMO System with Double Scattering Channels} \label{Sec:SystemModel}
We consider an uplink Massive MIMO system comprising $L$~cells, where cell~$l$ has one BS equipped with $M$ antennas and serving $K$ single-antenna users. Even though the propagation channels change over time and frequency, we use a quasi-static channel model where the time-frequency plane is divided into coherence blocks. Each coherence block comprises $\tau_c$ symbols such that the channel between an arbitrary user and the BS is static and frequency flat. This paper assumes that instantaneous channels are not known at the BSs. Therefore, in each coherence block, the $\tau_p$ symbols are dedicated to the pilot training phase and the remaining $\tau_c - \tau_p$ symbols are used for the uplink data transmission. The channel between user~$k$ in cell~$l$ and BS~$l'$ is modeled by the double scattering channel model \cite{yu2002models,van2016multi}, which is\footnote{This outdoor channel model was initiated for conventional MIMO systems under a far-field region and dedicated sub $6$-GHz bands for mobile services. In cellular Massive MIMO communications, the far-field effects are still observed since many antenna components can be practically installed in a small compact array \cite{bjornson2019massive}.}
\begin{equation} \label{eq:Channel}
\mathbf{h}_{lk}^{l'}  =  \sqrt{\frac{\beta_{lk}^{l'}}{S_{lk}^{l'}}} \left(\mathbf{R}_{lk}^{l'}\right)^{1/2} \mathbf{G}_{lk}^{l'} \left(\widetilde{\mathbf{R}}_{lk}^{l'} \right)^{1/2} \mathbf{g}_{lk}^{l'},
\end{equation}
where $\beta_{lk}^{l'}$ is the large-scale fading coefficient, which models the effects of the pathloss due to long distance and shadow fading due to obstacles. The integer parameter $S_{lk}^{l'}$ is the number of scatters generating the channel between BS~$l'$ and user~$k$ in cell~$l$. The matrix $\mathbf{R}_{lk}^{l'} \in \mathbb{C}^{M \times M }$ represents the correlation between the BS antennas and its scatterers; $\mathbf{G}_{lk}^{l'} \in \mathbb{C}^{M \times S_{lk}^{l'}}$ includes the small-scale fading coefficients between BS~$l'$ and its scattering cluster. The matrix $\widetilde{\mathbf{R}}_{lk}^{l'} \in \mathbb{C}^{S_{lk}^{l'} \times S_{lk}^{l'} }$ stands for the correlation between the transmit and receive
scatterers and $\mathbf{g}_{lk}^{l'} \in \mathbb{C}^{S_{lk}^{l'}}$ represents the small-scale fading
between the user and its scattering cluster. The elements of both $\mathbf{G}_{lk}^{l'}$ and $\mathbf{g}_{lk}^{l'}$ are independent and identically distributed as $\mathcal{CN}(0,1)$ by constraints on the trace of the covariance matrices. 
\begin{remark}
The double scattering channel model in \eqref{eq:Channel} reflects three important aspects of Massive MIMO channel propagation: the rank deficiency at the transceiver, the spatial fading correlation, and the signal attenuation by controlling multiple factors such as the number of scatterers in the environment, the correlation matrices, and the large-scale fading coefficients. It is more an involved channel model than in previous non-line-of-sight models to describe the sensitivity of the actual channel capacity to both the fading correlation and scattering structure in real propagation environments \cite{van2016multi, Gesbert2002a}. This model spans scenarios from uncorrelated Rayleigh to the single-keyhole channels. In practical systems, the covariance matrices can be estimated by averaging over many realizations of instantaneous channels, while the number of scatterers can be obtained by formulating and solving, for example, an $\ell_p-$norm optimization problem, which matches the double scattering channel model with measurement data \cite{patzold2003mobile}.
\end{remark} 
The further interesting statistical information of the double scattering channels, which is later useful for computing the uplink ergodic SE expression in a closed form, is presented in the following lemma.
\begin{lemma} \label{corollary:ChannelProperties}
	Let us consider the two random channel vectors $\mathbf{h}_{l'k'}^l$ and $\mathbf{h}_{l''k''}^l$ generated by the double scattering channel model and a deterministic matrix $\mathbf{B} \in \mathbb{C}^{ M \times M}$. If $ (l',k') \neq (l'', k'')$, it holds that
	\begin{equation} \label{eq:Cor1eq9}
		\mathbb{E} \left\{ \Big| \big(\mathbf{h}_{l'k'}^{l}\big)^{\rm H} \mathbf{B}  \mathbf{h}_{l''k''}^l \Big|^2 \right\} = b_{l'k'}^l \mathrm{tr} \big( \mathbf{B} \mathbf{R}_{l''k''}^l \mathbf{B}^{\rm H} \mathbf{R}_{l'k'}^l \big),
	\end{equation}
	with  $d_{l'k'}^l = \mathrm{tr}\big(\widetilde{\mathbf{R}}_{l'k'}^l\big)/S_{l'k'}^l $ and $b_{l'k'}^l = \beta_{l'k'}^l d_{l'k'}^l \beta_{l''k''}^l d_{l''k''}^l$. Moreover, for the channel $\mathbf{h}_{l'k'}^l$, it holds that
	\begin{multline} \label{eq:4moments}
		\mathbb{E} \left\{ \Big| \big(\mathbf{h}_{l'k'}^l \big)^{\rm H} \mathbf{B}  \mathbf{h}_{l'k'}^l \Big|^2 \right\} =  \big(\beta_{l'k'}^l \big)^2  \left( \big(d_{l'k'}^l \big)^2 + \frac{ \mathrm{tr} \Big( \left(\widetilde{\mathbf{R}}_{l'k'}^l\right)^2 \Big)}{\big(S_{l'k'}^l\big)^2} \right) \times \\
		\left( \Big| \mathrm{tr} \big( \mathbf{R}_{l'k'}^l  \mathbf{B}  \big)\Big|^2 + \mathrm{tr} \left( \mathbf{R}_{l'k'}^l \mathbf{B}   \mathbf{R}_{l'k'}^l \mathbf{B}^{\rm H} \right) \right).
	\end{multline}
\end{lemma}
\begin{proof}
	The proof is to compute the moments of non-Gaussian random variables and available in Appendix~\ref{Appendix:ChannelProperties}.
\end{proof}
In \eqref{eq:Cor1eq9}, the second moment obtained for the inner product of two different channel vectors is a deterministic value, which depends on their covariance matrices and scales up with the number of antennas installed at BS~$l$, say $M$. Meanwhile, the weighted forth moment in \eqref{eq:4moments} indicates a scaling factor of $M^2$. This moment is also inversely proportional to the number of scatterers. The moments of channels in Lemma~\ref{corollary:ChannelProperties} are utilized to compute the closed-form expression on the uplink ergodic rate of an arbitrary user.
\subsection{Uplink Pilot Training}
In each coherence block, each BS needs instantaneous channel state information for the uplink data detection. The $\tau_p$ symbols are dedicated to the uplink pilot training, which can create $\tau_p$ mutually orthogonal pilot signals. User~$k$ in cell~$l$ uses the deterministic pilot signal $\pmb{\phi}_{lk} \in \mathbb{C}^{\tau_p}$ with $\| \pmb{\phi}_{lk}\|^2 = \tau_p$. This pilot signal is also reused by other users in multiple cells and we can define the pilot reuse set as
\begin{equation} \label{eq:PilotReuseSet}
\mathcal{P}_{lk} = \left\{ (l',k'): \pmb{\phi}_{l'k'} = \pmb{\phi}_{lk}, l = 1, \ldots, L, k'=1, \ldots, K \right\},
\end{equation}
which contains the indices of all users sharing the same pilot signal as user~$k$ in cell~$l$, including $(l,k)$. Mathematically, it observes that
\begin{equation}
\pmb{\phi}_{lk}^{\rm H} \pmb{\phi}_{l'k'} = \begin{cases}
\tau_p, & \mbox{ if } (l',k') \in \mathcal{P}_{lk}, \\
0, & \mbox{ if } (l',k') \notin \mathcal{P}_{lk}.
\end{cases}
\end{equation}
At BS~$l$, the received pilot signal $\mathbf{Y}_l^p \in \mathbb{C}^{M \times \tau_p}$ with the superscript $p$ standing for the pilot training phase is formulated as
\begin{equation}
\mathbf{Y}_l^p  = \sum_{l'=1}^L \sum_{k'=1}^{K} \sqrt{\hat{p}_{l'k'}} \mathbf{h}_{l'k'}^l \pmb{\phi}_{l'k'}^{\rm H} + \mathbf{N}_l^p,
\end{equation}
where $\mathbf{N}_l^p \in \mathbb{C}^{M \times \tau_p}$ is additive noise with the independent and identically random elements distributed as $\mathcal{CN}(0, \sigma^2)$. BS~$l$ estimates the channel $\mathbf{h}_{l'k'}^l$ from user~$k'$ in cell~$l'$ by multiplying $\mathbf{Y}_{l}^p$ with the pilot sequence $\pmb{\phi}_{l'k'}$ as
\begin{equation} \label{eq:ylklp}
\mathbf{y}_{l'k'}^{l,p} = \mathbf{Y}_{l}^p \pmb{\phi}_{l'k'} = \sum_{(l'',k'') \in \mathcal{P}_{l'k'}} \sqrt{\hat{p}_{l''k''}} \tau_p \mathbf{h}_{l''k''}^l + \mathbf{N}_l^p \pmb{\phi}_{l'k'}.
\end{equation}
The minimum mean square error (MMSE) is not straightforward to apply to \eqref{eq:ylklp} because of the non Gaussian distributions. Nonetheless, the processed received signal $\mathbf{y}_{l'k'}^{l,p} \in \mathbb{C}^{M}$ has sufficient statistics to obtain a channel estimate of the origin $\mathbf{h}_{l'k'}^l$ by utilizing linear MMSE (LMMSE). We now consider the channel estimates under assumptions of statistical channel knowledge available at each BS.  
\begin{lemma} \label{lemma:ChannelEstPhaseAware}
By utilizing the LMMSE estimation, the channel estimate $\hat{\mathbf{h}}_{l'k'}^l \in \mathbb{C}^{M}$ from user~$k'$ in cell~$l'$ and BS~$l$ is 
\begin{equation} \label{eq:ChannelEst}
\hat{\mathbf{h}}_{l'k'}^l = \sqrt{\hat{p}_{l'k'}} \beta_{l'k'}^l d_{l'k'}^l \mathbf{R}_{l'k'}^l \pmb{\Psi}_{l'k'}^l  \mathbf{y}_{l'k'}^{l,p},
\end{equation}
where $\pmb{\Psi}_{l'k'}^l \in \mathbb{C}^{M \times M}$ is
\begin{equation}
\pmb{\Psi}_{l'k'}^l = \left(  \sum_{(l'',k'') \in \mathcal{P}_{l'k'}} a_{l''k''}^l \mathbf{R}_{l''k''}^l + \sigma^2 \mathbf{I}_{M}. \right)^{-1}.
\end{equation}
with $a_{l''k''}^l =  \tau_p \hat{p}_{l''k''} \beta_{l''k''}^l d_{l''k''}^l$. The covariance matrix of the channel estimate $\hat{\mathbf{h}}_{l'k'}^l$ is computed as
\begin{equation} \label{eq:CovarianceEst}
\mathbb{E} \Big\{ \hat{\mathbf{h}}_{l'k'}^l \big(\hat{\mathbf{h}}_{l'k'}^l\big)^{\rm H} \Big\} =\hat{p}_{l'k'} \big(\beta_{l'k'}^l\big)^2 \big(d_{l'k'}^l \big)^2 \tau_p \mathbf{R}_{l'k'}^l \pmb{\Psi}_{l'k'}^l  \mathbf{R}_{l'k'}^l.
\end{equation}
\end{lemma}
\begin{proof}
The proof is based on the LMMSE estimation of non-Gaussian random variables \cite{Kay1993a}, but adapted to our framework with the channel vector in \eqref{eq:Channel} and the pilot reuse in \eqref{eq:PilotReuseSet}. The detail proof is available in Appendix~\ref{appendix:ChannelEstPhaseAware}.
\end{proof}
Lemma~\ref{lemma:ChannelEstPhaseAware} shows the concrete expression of the channel estimate of each user together with the statistical information, which are used to formulate the combining vectors and computing the closed-form expression on the uplink ergodic SE hereafter. It should be noticed that our channel estimation considers the influence of coherent interference caused by the pilot contamination in multi-cell Massive MIMO scenarios, which is a generalization of the previous result in \cite{kammoun2019asymptotic,Ye2020} that assumed the orthogonal pilot signals for all the users in a single cell.  Along with the statistical information in Lemma~\ref{corollary:ChannelProperties}, the channel estimates and estimation errors in Lemma~\ref{lemma:ChannelEstPhaseAware} are utilized to compute the closed-form uplink SE expression hereafter. 
\subsection{Uplink Data Transmission}
During the uplink data transmission, user~$k$ in cell~$l$ sends a data symbol $s_{lk}$ with $\mathbb{E} \{ |s_{lk}|^2\} = 1$ and the received data signal $\mathbf{y}_l \in \mathbb{C}^{M}$ at BS~$l$ is a superposition of all the transmitted signals from all the users as
\begin{equation}
\mathbf{y}_l = \sum_{l'=1}^L \sum_{k'=1}^{K} \sqrt{p_{l'k'}} \mathbf{h}_{l'k'}^l s_{l'k'} + \mathbf{n}_l,
\end{equation}
where $p_{l'k'}$ is the transmit power of user~$k'$ in cell~$l'$ assigned to each data symbol and $\mathbf{n}_l$ is additive noise distributed as $\mathcal{CN}(\mathbf{0}, \sigma^2 \mathbf{I}_{M})$. By utilizing a combining vector $\mathbf{v}_{lk} \in \mathbb{C}^{M}$ based on the channel estimates, BS~$l$ decodes the desired signal from user~$k$ in cell~$l$ as
\begin{multline}
\mathbf{v}_{lk}^H \mathbf{y}_l =  \sqrt{p_{lk}} \mathbb{E} \big\{ \mathbf{v}_{lk}^H  \mathbf{h}_{lk}^l \big\} s_{lk} + \sqrt{p_{lk}} \Big( \mathbf{v}_{lk}^H  \mathbf{h}_{lk}^l - \mathbb{E} \big\{ \mathbf{v}_{lk}^H  \mathbf{h}_{lk}^l \big\} \Big) s_{lk} \\
+ \sum_{\substack{k'=1, k' \neq k}}^{K} \mathbf{v}_{lk}^H \sqrt{p_{lk'}} \mathbf{h}_{lk'}^l s_{lk'} \\
+ \sum_{\substack{l'=1, l' \neq l}}^L \sum_{k'=1}^{K} \mathbf{v}_{lk}^H \sqrt{p_{l'k'}} \mathbf{h}_{l'k'}^l s_{l'k'} + \mathbf{v}_{lk}^H \mathbf{n}_l,
\end{multline}
where the first term contains the desired signal by virtue of the channel hardening \cite{Marzetta2016a}. The second term describes the beamforming uncertainty effects, while the remaining terms are mutual interference and noise. As shown in \cite{massivemimobook}, the uplink ergodic SE is obtained by the use-and-then-forget channel capacity bounding technique as
\begin{equation} \label{eq:AchievableRate}
R_{lk} = \left(1 - \frac{\tau_p}{\tau_c}\right) \log_2 \left( 1 + \mathrm{SINR}_{lk} \right), [\mbox{b/s/Hz}],
\end{equation}
where the effective signal-to-interference-and-noise ratio (SINR) value is computed as in \eqref{eq:GeneralSINR}.
\begin{figure*}
\begin{equation} \label{eq:GeneralSINR}
\mathrm{SINR}_{lk} = \frac{p_{lk} \Big| \mathbb{E} \big\{ \mathbf{v}_{lk}^H \mathbf{h}_{lk}^l \big\} \Big|^2 }{ \sum_{l'=1}^L \sum_{k'=1}^{K} p_{l'k'} \mathbb{E} \Big\{ \big| \mathbf{v}_{lk}^H \mathbf{h}_{l'k'}^l \big|^2 \Big\} -  p_{lk} \Big| \mathbb{E} \big\{ \mathbf{v}_{lk}^H \mathbf{h}_{lk}^l \big\} \Big|^2  + \sigma^2 \mathbb{E} \{ \|\mathbf{v}_{lk} \|^2 \}}.
\end{equation}
\hrulefill
\end{figure*}
The expectations in \eqref{eq:GeneralSINR} are taking over all the sources of randomness and \eqref{eq:AchievableRate} is an achievable rate since it is a lower bound on the channel capacity. Furthermore, this achievable rate can be computed numerically for any combining scheme. The main demerit of \eqref{eq:AchievableRate} is high computational complexity since many instantaneous channels need to be gathered such that several expectations can be numerically estimated.
\subsection{Uplink Spectral Efficiency Analysis}
If MR combining is used by each BS, i.e.,$\big(\mathbf{v}_{lk} = \hat{\mathbf{h}}_{lk}^l\big),\forall l,k$, we obtain the closed-form expression for the uplink SE in \eqref{eq:AchievableRate} as shown by Theorem~\ref{Theorem:ClosedFormMR}.\footnote{The framework in this paper can be easily extended to the downlink data transmission.}
\begin{theorem} \label{Theorem:ClosedFormMR}
When BS~$l$ uses the MR combing vector to decode the desired signal from user~$k$ in cell~$l$, the achievable uplink SE obtained in \eqref{eq:AchievableRate} with the closed-form expression of the SINR value computed as
\begin{equation} \label{eq:SINRlkMR}
\mathrm{SINR}_{lk} = \frac{ p_{lk} z_{lk}^l \Big| \mathrm{tr} \big(\mathbf{R}_{lk}^l \pmb{\Psi}_{lk}^l \mathbf{R}_{lk}^l \big) \Big|^2 }{ \mathsf{NI}_{lk} + \mathsf{CI}_{lk} + \mathsf{NO}_{lk} },
\end{equation}
where $\mathsf{NI}_{lk}, \mathsf{CI}_{lk},$ and $\mathsf{NO}_{lk}$ are respectively the non-coherent interference, coherent interference, and noise, which are computed in the closed-form expression as
\begin{align}
&\mathsf{NI}_{lk} = \sum_{l' =1 }^L \sum_{k'=1}^{K}  p_{l'k'} m_{l'k'}^l   \mathrm{tr} \big( \mathbf{R}_{lk}^l \pmb{\Psi}_{lk}^l \mathbf{R}_{lk}^l \mathbf{R}_{l'k'}^l  \big), \label{eq:NIlk} \\
&\mathsf{CI}_{lk} =  \sum_{(l',k') \in \mathcal{P}_{lk} \setminus (l,k)} p_{l'k'} z_{l'k'}^l \Big| \mathrm{tr} \big(  \mathbf{R}_{l'k'}^l \pmb{\Psi}_{lk}^l  \mathbf{R}_{lk}^l  \big)\Big|^2  + \notag  \\
& \sum_{(l',k') \in \mathcal{P}_{lk}} p_{l'k'}     \frac{z_{l'k'}^l \mathrm{tr} \Big(\big(\widetilde{\mathbf{R}}_{l'k'}^l\big)^2 \Big)}{\big(d_{l'k'}^l S_{l'k'}^l\big)^2} \Big| \mathrm{tr} \big(  \mathbf{R}_{l'k'}^l \pmb{\Psi}_{lk}^l  \mathbf{R}_{lk}^l  \big)\Big|^2  + \notag \\
& \sum_{(l',k') \in \mathcal{P}_{lk}} p_{l'k'} z_{l'k'}^l  \frac{ \mathrm{tr} \Big(\big(\widetilde{\mathbf{R}}_{l'k'}^l \big)^2 \Big)}{\big(S_{l'k'}^l\big)^2} \mathrm{tr} \left(  \mathbf{R}_{l'k'}^l \pmb{\Psi}_{lk}^l  \mathbf{R}_{lk}^l  \mathbf{R}_{l'k'}^l \mathbf{R}_{lk}^l \pmb{\Psi}_{lk}^l   \right),  \label{eq:CIlk}\\
&\mathsf{NO}_{lk} =  \sigma^2 \hat{p}_{lk} \big(\beta_{lk}^l \big)^2 \big(d_{lk}^l \big)^2 \tau_p \mathrm{tr} \big(\mathbf{R}_{lk}^l \pmb{\Psi}_{lk}^l \mathbf{R}_{lk}^l \big),
\end{align}
with the values $m_{l'k'}^l$ and $z_{l'k'}^l, \forall l',k',l,$ defined as
\begin{align}
&m_{l'k'}^l = \beta_{l'k'}^l d_{l'k'}^l \hat{p}_{lk} \big(\beta_{lk}^l \big)^2 \big(d_{lk}^l \big)^2 \tau_p,\\
&z_{l'k'}^l = \hat{p}_{l'k'}  \big(\beta_{l'k'}^l \big)^2 \big(d_{l'k'}^l \big)^2 \hat{p}_{lk} \big(\beta_{lk}^l \big)^2 \big( d_{lk}^l \big)^2 \tau_p^2.
\end{align}
\end{theorem}
\begin{proof}
The proof is obtained by computing the expectations of non-Gaussian random variables in \eqref{eq:GeneralSINR}. The detailed proof is available in Appendix~\ref{Appendix:ClosedFormMR}.
\end{proof}
The SINR expression \eqref{eq:SINRlkMR} is explicitly influenced by many factors such as channel covariance matrices, the number of scatters, pilot reuse, channel estimation quality, which are hidden in the general formulation \eqref{eq:GeneralSINR}. Specifically, the numerator of  \eqref{eq:SINRlkMR} shows the contribution of both channel estimation quality and covariance matrix of user~$k$ in cell~$l$.  Moreover, the effectiveness of the array gain is verified since the numerator scales up with the number of antennas thanks to the spatial covariance property in \eqref{eq:Assump1}. The first part in the denominator of \eqref{eq:SINRlkMR} demonstrates the degradation of the received signal quality due to non-coherent interference. The second part presents the contributions of coherent interference caused by reusing the pilot signals among the users that is defined by the pilot reuse set $\mathcal{P}_{lk}$. Unlike previous works with many scatterers \cite{Chien2018a}, this part also shows that a small number of scatterers have significant contributions to increasing non-coherent interference. If the coherent blocks are large enough such that pilot sequences allocated to all users are pairwisely orthogonal, i.e., $\tau_p \geq LK$, the SINR value of user~$k$ in cell~$l$ is still computed as \eqref{eq:SINRlkMR}, but the following parameters are reformulated as
\begin{align}
\pmb{\Psi}_{lk}^l &= \left(  a_{lk}^l \mathbf{R}_{lk}^l + \sigma^2 \mathbf{I}_{M} \right)^{-1},\\
\mathsf{CI}_{lk} &= p_{lk}     \frac{z_{lk}^l \mathrm{tr} \Big(\big(\widetilde{\mathbf{R}}_{lk}^l\big)^2 \Big)}{\big(d_{lk}^l S_{lk}^l\big)^2} \Big| \mathrm{tr} \big(  \mathbf{R}_{lk}^l \pmb{\Psi}_{lk}^l  \mathbf{R}_{lk}^l  \big)\Big|^2  + \notag \\
&  p_{lk} z_{lk}^l  \frac{ \mathrm{tr} \Big(\big(\widetilde{\mathbf{R}}_{lk}^l \big)^2 \Big)}{\big(S_{lk}^l\big)^2} \mathrm{tr} \left(  \mathbf{R}_{lk}^l \pmb{\Psi}_{lk}^l  \mathbf{R}_{lk}^l  \mathbf{R}_{lk}^l \mathbf{R}_{lk}^l \pmb{\Psi}_{lk}^l   \right), 
\end{align}
which demonstrates the influences of a finite scatterer number to the uplink SE. Finally, the last part in the denominator of \eqref{eq:SINRlkMR} presents the additive noise effects.
\begin{remark}
We consider the MR combining technique  due to its low computational complexity. This linear combining technique allows the execution of SE analysis in the closed form with a finite set of BS antennas, users, and scatterers. In addition, it can be implemented by only using the local channel state information, and therefore, easy to implement in a distributed manner.
\end{remark}
\subsection{Asymptotic Analysis}
In order to observe the uplink SE at an asymptotic regime and also compare with previous works, we now investigate the uplink asymptotic SE of each user when $M \rightarrow \infty$ and $S_{l'k'}^l \rightarrow \infty, \forall l,l',k'$. Aligned with previous works \cite{van2020massivechannels}, the general preliminary settings on the covariance matrices are given in Assumption~\ref{Assumption1}. 
\begin{assumption} \label{Assumption1}
For $l,l'=1, \ldots, L$ and $k'=1, \ldots, K,$ the spatial covariance matrices $\mathbf{R}_{l'k'}^l$ and $\widetilde{\mathbf{R}}_{l'k'}^l$ satisfy
\begin{align}
& \limsup_{M} \big\| \mathbf{R}_{l'k'}^l  \big\|_2 <  \infty, \, \, \liminf_{M} \frac{ \mathrm{tr} \big(\mathbf{R}_{l'k'}^l \big)}{ M} > 0, \label{eq:Assump1}\\
& \limsup_{S_{l'k'}^l} \big\| \widetilde{\mathbf{R}}_{l'k'}^l  \big\|_2 <  \infty, \, \, \liminf_{S_{l'k'}^l} \frac{ \mathrm{tr} \big(\widetilde{\mathbf{R}}_{l'k'}^l \big)}{ S_{l'k'}^l} > 0.
\end{align}
\end{assumption}
Assumption~\ref{Assumption1} is established based on the fact that a double scattering channel has two covariance matrices on the definition. This assumption is extended from the standard form in the asymptotic analysis for Massive MIMO communications with a single covariance matrix \cite{Hoydis2013a}. Physically, the gathered signal energy at a BS originates from many spatial directions and is proportional to the number of antennas. We also utilize the spatial orthogonality between two covariance matrices to seek for a convergence point at the asymptotic regime as shown in Definition~\ref{Def:Orthogonal}.
\begin{definition} \label{Def:Orthogonal}
The two covariance matrices $\mathbf{R}_{l'k'}^l$ and $\mathbf{R}_{lk}^l, \forall l',l,k,k'$ are asymptotically spatially orthogonal if
\begin{equation} \label{eq:Orthogonal}
\frac{1}{M} \mathrm{tr} \Big( \mathbf{R}_{l'k'}^l \mathbf{R}_{lk}^l \Big) \rightarrow 0, M \rightarrow \infty.
\end{equation}
\end{definition}
As pointed out in  previous works \cite{Bjornson2017bo, Yin2013a}, the condition \eqref{eq:Orthogonal} indicates the two users having orthogonal correlation eigenspaces. This holds for a network where each BS is equipped with antennas in a uniform linear array and the supports of the multi-path angular distributions of the two users are strictly non-overlapping. The convergence of the uplink SE for each user is stated in Theorem~\ref{theorem:Asymptotic}.
\begin{theorem} \label{theorem:Asymptotic}
Under Assumption~\ref{Assumption1}, the uplink SE of user~$k$ in cell~$l$ can be asymptotically observed by the following cases:
\begin{itemize}
\item[a)] As $M \rightarrow \infty$ and a given set of finite scatterers, the achievable rate of user~$k$ in cell~$l$ converges to
\begin{multline} \label{eq:AsymptRatev1}
R_{lk} = \left( 1 - \frac{\tau_p}{\tau_c} \right) \times \\ \log_2 \left( 1 + \frac{ p_{lk} z_{lk}^l \Big| \mathrm{tr} \big(\mathbf{R}_{lk}^l \pmb{\Psi}_{lk}^l \mathbf{R}_{lk}^l \big) \Big|^2 }{ \mathsf{CI}_{lk} } \right), [\mbox{b/s/Hz}].
\end{multline}
\item[\textit{b)}]   As $M \rightarrow \infty$, a limited number of scatterers, and the two covariance matrices $\mathbf{R}_{l'k'}^l$ and $\mathbf{R}_{lk}^l$ are asymptotically orthogonal for all $(l',k') \in \mathcal{P}_{lk} \setminus (l,k)$, the achievable rate of user~$k$ in cell~$l$ converges to
\begin{equation} \label{eq:Limitb}
R_{lk} = \left(1 - \frac{\tau_p}{\tau_c} \right) \log_2 \left( 1 + \frac{ \big( d_{lk}^l S_{lk}^l\big)^2 }{ \mathrm{tr} \Big( \big( \widetilde{\mathbf{R}}_{l'k'}^l \big)^2 \Big)} \right),  [\mbox{b/s/Hz}].
\end{equation}
\item[\textit{c)}]   As $M \rightarrow \infty$ and $S_{l'k'}^l \rightarrow \infty, \forall l',k' \in \mathcal{P}_{lk},$ the achievable rate of user~$k$ in cell~$l$ converges to
\begin{multline} \label{eq:AsymptRatev3}
R_{lk} = \left( 1 - \frac{\tau_p}{\tau_c} \right) \times \\
 \log_2 \left( 1 + \frac{ p_{lk} z_{lk}^l \Big| \mathrm{tr} \big(\mathbf{R}_{lk}^l \pmb{\Psi}_{lk}^l \mathbf{R}_{lk}^l \big) \Big|^2 }{ \widetilde{\mathsf{CI}}_{lk} } \right),  [\mbox{b/s/Hz}],
\end{multline}
where $\widetilde{\mathsf{CI}}_{lk} = \sum_{(l',k') \in \mathcal{P}_{lk} \setminus (l,k)} p_{l'k'} z_{l'k'}^l \left| \mathrm{tr} \left(  \mathbf{R}_{l'k'}^l \pmb{\Psi}_{lk}^l  \mathbf{R}_{lk}^l  \right)\right|^2 $.
\item[\textit{d)}] As $M \rightarrow \infty$,  $S_{l'k'}^l \rightarrow \infty, \forall l',k' \in \mathcal{P}_{lk}$, and the two covariance matrices $\mathbf{R}_{l'k'}^l$ and $\mathbf{R}_{lk}^l$ are asymptotically orthogonal for all $(l',k') \in \mathcal{P}_{lk} \setminus (l,k)$, the achievable rate of user~$k$ in cell~$l$ grows without bound as
\begin{equation} \label{eq:AsymptRatev4}
R_{lk} \rightarrow \infty, [\mbox{b/s/Hz}].
\end{equation}
\end{itemize}
\end{theorem}
\begin{proof}
The proof is to compute the asymptotic SE of each user in the network with Assumption~\ref{Assumption1} and Definition~\ref{Def:Orthogonal} when the number of antennas and/or scatterers increases. The detailed proof is available in Appendix~\ref{Appendix:Asymptotic}.
\end{proof}	
Theorem~\ref{theorem:Asymptotic} reveals that the uplink SE at an asymptotic regime is dependent on both the number of antennas at each BS and scatterers in propagation environments as well. For a limited number of scatterers at each communication link, the uplink SE of user~$k$ in cell~$l$ is bounded when the number of antennas increases due to the pilot contamination effects. Different from \cite{ozdogan2019massive}, the SE converges to a finite point as shown \eqref{eq:Limitb} even when the asymptotically orthogonality among covariance matrices holds because of lacking the scatterers. For a rich scattering environment, the limitation is mainly from reusing the pilot signals among users causing coherent interference, which is dominant at an asymptotic regime. The fundamental difference of the double scattering channels compared with other spatial fading models as correlated Rayleigh fading or local scattering fading is that the unbounded channel capacity is obtained when the covariance matrices are asymptotically orthogonal as well as both numbers of antennas at each BS and scatterers go asymptotically.

\section{Uplink Total Data Energy Consumption Minimization} \label{Sec:Optimization}
This section expresses an uplink energy consumption minimization problem by assuming that user~$k$ in cell~$l$ requests a SE $\xi_{lk} > 0, \forall l,k,$ and has a maximum power $P_{\max,lk} >0$. Investigating this optimization problem, we further manifest the feasibility for user locations, where all the users are served with the requested SE under the limited power budget. In contrast, the infeasibility is manifested for certain user locations, where users may be served with the SE lower than what has been requested.
\subsection{Problem Formulation}
The main goal of 5G-and-beyond systems is to provide the high SEs to all users with a minimal power consumption. In this paper, we formulate a total data energy optimization problem for the uplink data transmission as follows  
\begin{equation} \label{Prob:TotalPowerOptv1}
	\begin{aligned}
		& \underset{\{ p_{lk} \geq 0 \} }{\mathrm{minimize}}
		&&   (\tau_c - \tau_p) \sum_{ l=1}^L \sum_{k=1}^K p_{lk} \\
		& \,\,\mathrm{subject \,to}
		& &  R_{lk} \geq \xi_{lk}, \forall l,k, \\
		& & & p_{lk} \leq P_{\max,lk}, \forall l,k,
	\end{aligned}
\end{equation}
where $ P_{\max,lk}$ is the maximum power level that user~$k$ in cell~$l$ can allocate to each data symbol. Problem~\eqref{Prob:TotalPowerOptv1} constrains on the rate requirement and limited power budget of each user. The per-user power constraints implicitly indicate that the total transmit power in the network should be upper bounded. In addition, the objective function of problem~\eqref{Prob:TotalPowerOptv1} ensures the minimal network power consumption. Therefore, our proposed optimization problem is able to reduce the mutual interference on other networks. 
\begin{remark}  
Note that, in \eqref{Prob:TotalPowerOptv1}, we consider the per-user power constraints. It is also interesting to additionally consider a network power constraint so that the mutual interference on other networks can be controlled more effectively. For this case, the feasibility of our optimization problem is a main issue. We may first check if the network power constraint would be active in the selected point, i.e., if the network power constraint is satisfied under the optimized individual constraints. If it is inactive, the solution remains unaffected. If it is active, a heuristic approach would be to reduce the number of users, increase the number of antennas, or relax the per-user SE requirements. This potential extension is left for the future work. In this paper, we assume that the network power constraint is always satisfied and only handling a scenario that the per-user powers are constrained.
\end{remark}
By setting $\nu_{lk} = 2^{\xi_{lk} \tau_c /(\tau_c - \tau_p)} - 1$ and removing the constant $\tau_c - \tau_p$ in the objective function, problem~\eqref{Prob:TotalPowerOptv1} is converted from the SE constraints into the equivalent SINR constraints as
\begin{equation} \label{Prob:TotalPowerOptv2}
	\begin{aligned}
		& \underset{\{ p_{lk} \geq 0 \} }{\mathrm{minimize}}
		&&    \sum_{ l=1}^L \sum_{k=1}^K p_{lk} \\
		& \,\,\mathrm{subject \,to}
		& &  \mathrm{SINR}_{lk} \geq \nu_{lk}, \forall l,k, \\
		& & & p_{lk} \leq P_{\max,lk}, \forall l,k.
	\end{aligned}
\end{equation}
Instead of optimizing the energy consumption as \eqref{Prob:TotalPowerOptv1}, problem~\eqref{Prob:TotalPowerOptv2} minimizes the total transmit powers, which all users consume for the uplink data transmission. Due to the universe of all SINR expressions $\{ \mathrm{SINR}_{lk} \}$, problem~\eqref{Prob:TotalPowerOptv2} is in a general form for any combining technique. We now focus on MR combining technique as the corresponding SINRs have been derived in closed-form as obtained in Theorem~\ref{Theorem:ClosedFormMR}. The concrete optimization problem is reformulated by utilizing the SINR expression \eqref{eq:SINRlkMR} into \eqref{Prob:TotalPowerOptv2} as
\begin{equation} \label{Prob:TotalPowerOptv3}
	\begin{aligned}
		& \underset{\{ p_{lk} \geq 0 \} }{\mathrm{minimize}}
		&&    \sum_{ l=1}^L \sum_{k=1}^K p_{lk} \\
		& \,\,\mathrm{subject \,to}
		& &   \frac{ p_{lk} z_{lk}^l \left| \mathrm{tr} \left(\mathbf{R}_{lk}^l \pmb{\Psi}_{lk}^l \mathbf{R}_{lk}^l \right) \right|^2 }{ \mathsf{NI}_{lk} + \mathsf{CI}_{lk} + \mathsf{NO}_{lk}} \geq \nu_{lk}, \forall l,k, \\
		& & & p_{lk} \leq P_{\max,lk}, \forall l,k.
	\end{aligned}
\end{equation}
We stress that problem~\eqref{Prob:TotalPowerOptv3} jointly optimizes the powers to satisfy the requested SINRs from all the users. The required SINR levels $\nu_{lk}, \forall l,k,$ are distinct from each other in practice and the global optimum is only found when all the users are simultaneously served by the required SEs. This problem can be either feasible or infeasible for a given set of user locations and shadow fading realizations as presented hereafter.\footnote{The congestion issue may appear in the other optimization problems as the spectral or energy efficiency maximization subject to the SE requirements and/nor the limited power budget constraints. The key argument of our framework is to point out that many users might still be served with their SE requirements in Massive MIMO communications if there is a strategic policy to deal with a few unsatisfied users.}
\subsection{Feasible and Infeasible Problems} \label{SubSec:FeasInfeas}
When problem~\eqref{Prob:TotalPowerOptv3} has a non-empty feasible set meaning that the network is able to simultaneously provide the required SEs to all the users conditioned on the power constraints. We can find the global optimal solution to  problem~\eqref{Prob:TotalPowerOptv3}. Indeed, the objective function is a linear combination of all the power variables $\{ p_{lk}\}, \forall l,k$. In addition, the power budget constraint functions are affine while the SINR constraints, $\forall l,k,$ are reformulated as
\begin{equation}
	\nu_{lk} \mathsf{NI}_{lk} + \nu_{lk} \mathsf{CI}_{lk} + \nu_{lk} \mathsf{NO}_{lk} \leq  p_{lk} z_{lk}^l \left| \mathrm{tr} \left(\mathbf{R}_{lk}^l \pmb{\Psi}_{lk}^l \mathbf{R}_{lk}^l \right) \right|^2, 
\end{equation}
which are also affine functions. Consequently, \eqref{Prob:TotalPowerOptv3} is a linear program on standard form \cite{Boyd2004a}. We hence enable to solve \eqref{Prob:TotalPowerOptv3} to the global optimality in polynomial time, for instance, utilizing a general interior-point optimization toolbox as CVX \cite{cvx2015}. Problem~\eqref{Prob:TotalPowerOptv3} includes the $KL$ optimization variables and the $2KL$ constraints and as such it has the computational complexity of the order $\mathcal{O}\left( N_i 2K^3L^3 \right)$, where $N_i$ is the number of Newton iterations needed to obtain a predetermined precision, typically in the order of tens \cite[Chapter 11]{Boyd2004a}. It should be noticed that all the $KL$ users will spend non-zero data powers at the global optimum when problem~\eqref{Prob:TotalPowerOptv3} is feasible owning to the  non-zero SE requirements.

For a specific realization of user locations and the power budgets, there may be a situation that all the users cannot be simultaneously served by the SE requirements. We emphasize that only one unfortunate user served with a lower SE suffices to create an empty feasible domain for the total transmit power optimization problem. Alternatively, problem~\eqref{Prob:TotalPowerOptv3} lacks a feasible solution \cite[Section 4.1]{Boyd2004a}. The unsatisfied SE is caused by high mutual interference in cellular networks and/or extreme locations as the cell edge leading to some users having a weak channel. Moreover, a user may require a too high SE for which the system cannot provide this service even spending maximum data power. Fortunately, a feasible solution of the data powers might still exist for most of the users with the required SEs, while only one or a few users are \textit{unsatisfied}. Consequently, it may be sufficient to remove or reduce the required SEs of those unsatisfied users to convert an infeasible problem to a feasible one. However, it is not trivial to identify which users are unsatisfied to completely remove during solving problem~\eqref{Prob:TotalPowerOptv3}. As one of the main contributions, this paper develops the power allocation strategies to handle such infeasible instances by allowing the corresponding SINR constraints to be violated.

\section{Congestion solution based on alternating optimization} \label{Sec:Solutions}
This section proposes the two algorithms attaining a fixed-point solution to problem~\eqref{Prob:TotalPowerOptv3} with either empty or non-empty feasible set. When the feasible set is empty, the SINR constraints of users, which potentially make the congestion issue are relaxed: The first approach is spending the maximum power on unsatisfied users. In contrast, the second approach is reducing the data power of those unsatisfied users. We now introduce important notations which will be widely utilized in this paper to construct the proposed algorithms as shown in Definition~\ref{Def2}.
\begin{definition} \label{Def2}
Let us denote $\mathbf{z}$ and $\mathbf{z}'$ the real vectors of size $KL \times 1$, for which the $n$-th elements are $z_n$ and $z_n'$, respectively. The notation $\mathbf{z} \succeq \mathbf{z}'$ indicates element-wise inequality $z_n \geq z_n', \forall n = 1, \ldots, KL$. Meanwhile, the notation $\mathbf{z} \preceq  \mathbf{z}'$ indicates $z_n \leq z_n', \forall n = 1, \ldots, KL$.
\end{definition}
\subsection{Spending Maximum Transmit Power on Unsatisfied Users}
For the glorification of simplification in comprehension, problem~\eqref{Prob:TotalPowerOptv3} with a non-empty feasible domain is first considered. We stack all the data powers into a vector $\mathbf{p} = [p_{11}, \ldots, p_{LK}]^T \in \mathbb{R}_{+}^{LK}$, then the SINR constraint of user~$k$ in cell~$l$ is reformulated as
\begin{equation}
p_{lk} \geq I_{lk} (\mathbf{p}),
\end{equation}
where $I_{lk} (\mathbf{p})$ is so-called a standard interference function, which is given by
\begin{equation} \label{eq:Ilk}
I_{lk} (\mathbf{p}) = \frac{\nu_{lk} \mathsf{NI}_{lk} (\mathbf{p}) + \nu_{lk} \mathsf{CI}_{lk} (\mathbf{p}) + \nu_{lk} \mathsf{NO}_{lk} }{ z_{lk}^l \left| \mathrm{tr} \left(\mathbf{R}_{lk}^l \pmb{\Psi}_{lk}^l \mathbf{R}_{lk}^l \right) \right|^2 }.
\end{equation}
In \eqref{eq:Ilk}, the detailed expressions of $\mathsf{NI}_{lk} (\mathbf{p})$ and $\mathsf{CI}_{lk} (\mathbf{p})$ have been already expressed in \eqref{eq:NIlk} and \eqref{eq:CIlk}, but we here emphasize them as the functions of data power variables stacked in $\mathbf{p}$. We now introduce the definition of a standard interference function for which an low complexity algorithm to obtain a fixed point solution is proposed.
\begin{definition}[Standard interference function] \label{Def:TypeI}
A function $I(\mathbf{z})$ is a standard interference function for all $\mathbf{z} \succeq \mathbf{0}$, if the following properties hold:
$a)$ Positivity $I(\mathbf{z}) >0, \forall \mathbf{z} \succeq 0$. $b)$ Monotonicity $I(\mathbf{z})  \geq I(\mathbf{z}')$ if $\mathbf{z} \succeq \mathbf{z}'$. $c)$ Scalability: $\alpha I(\mathbf{z}) > I(\alpha \mathbf{z}), \forall \alpha > 1,$ for all scalar $\alpha > 1$.
\end{definition}
The positivity property is because of the inherent mutual interference and thermal noise in the system, which implies a non-zero value. This means that the transmit data powers are always larger than zero when users request non-zero SEs. The monotonicity property ensures that we can scale up or down \eqref{eq:Ilk} by adjusting the data powers. Finally, the scalability property suggests a method to uniformly scale down the data power coefficient of user~$k$ in cell~$l$ 
at each iteration by utilizing a positive constant $\alpha$. We now construct a policy to update the data power of user~$k$ in cell~$l$ for the given initial values $p_{lk}(0), \forall l,k,$ as in Theorem~\ref{Theorem:Alg1}.
\begin{theorem}\label{Theorem:Alg1}
By assuming that the feasible domain is non-empty and $0 \leq I_{lk} (\mathbf{p}) \leq P_{\max,lk}^2$ always holds for all $\mathbf{p}$ in the feasible domain. For the initial values of data powers $p_{lk} (0) = P_{\max,lk}, \forall l,k$, there exist data powers for which each interference function $I_{lk} (\mathbf{p})$ is non-increasing along iterations and converges to a fixed point. Particularly, the data power of user~$k$ in cell~$l$, denoted by $p_{lk}(n)$, can be updated at iteration~$n$ as
\begin{equation} \label{eq:plkn}
p_{lk}(n) = I_{lk} (\mathbf{p}(n-1)), \forall l,k.
\end{equation}
\end{theorem}
\begin{proof}
The proof is to testify every function $I_{lk} (\mathbf{p})$ defined in \eqref{eq:Ilk} being standard interference, and hence the updated power policy in \eqref{eq:plkn} ensures that this iterative approach will converge to a fixed point. The detailed proof is available in Appendix~\ref{Appendix:Alg1}. 
\end{proof}
Every user in the network has its own standard interference function satisfying the three fundamental properties in Definition~\ref{Def:TypeI} and utilizing it to update the data power as in \eqref{eq:plkn}. The analysis in Theorem~\ref{Theorem:Alg1} is based on the assumption that problem~\eqref{Prob:TotalPowerOptv3} has the global optimum for which all users are served with their required SEs. The power constraints in \eqref{Prob:TotalPowerOptv3} ($p_{lk} \leq P_{\max, lk}, \forall l,k$) are tackled by the fact if $I_{lk}(n-1) > P_{\max,lk}$, then the congestion issue appears and leads to an obvious selection $p_{lk} (n) =  P_{\max,lk}$. We therefore define the constrained standard interference function used at iteration~$n-1$ as
\begin{equation} \label{eq:Ihatlk}
\hat{I}_{lk}(\mathbf{p}(n-1)) = \min \left( I_{lk}(\mathbf{p}(n-1)), P_{\max,lk}   \right).
\end{equation}
For a cellular Massive MIMO system with the power budget constraints and the initial data power vector $\mathbf{p} (0)$ with the entries $p_{lk}(0) = P_{\max,lk}, \forall l,k,$ iteration~$n$ updates the data power of user~$k$ in cell~$l$ as
\begin{equation} \label{eq:plkv1}
p_{lk}(n) = \hat{I}_{lk}(\mathbf{p}(n-1)).
\end{equation}
Combining \eqref{eq:Ihatlk} and \eqref{eq:plkv1}, we observe that if $\hat{I}_{lk}(\mathbf{p}(n-1)) = P_{\max,lk}$, the update $p_{lk}(n) = P_{\max,lk}$ maintains the non-increasing objective function of problem~\eqref{Prob:TotalPowerOptv3}. Otherwise, it holds that $\hat{I}_{lk}(\mathbf{p} (n-1)) = I_{lk}(\mathbf{p} (n-1))$, and hence user~$k$ in cell~$l$ consumes less power than the maximum. This procedure will be applied to all the $KL$ users, which results in an alternating approach is summarized in Algorithm~\ref{Algorithm1}. Since the convergence of the update $p_{lk}(n) = P_{\max,lk}$ is trivial, the proposed algorithm converges to a fixed point follows a similar methodology as \cite[Theorem 7]{Yates1995a}. By assuming that the channel statistic information is computed in advance and available in the network, we can compute the total number of operations that dominate the computational complexity of this algorithm as
$ \mathcal{O} \left(N_mL^2K^2 + 3 N_m \left| \mathcal{P}_{lk} \right|LK \right),$ where $N_m$ is the number of iterations needed to reach the fixed point in polynomial time. Notice that, in Algorithm~\ref{Algorithm1}, when users cannot be served by the
required SEs, one still lets them utilize the maximum power. This policy aims at maximizing the SE of a particular user, however producing more mutual interference to the other users.

\begin{algorithm}[t]
	\caption{Data power allocation to problem~\eqref{Prob:TotalPowerOptv3} by spending maximum transmit power on unsatisfied users} \label{Algorithm1}
	\textbf{Input}:  Define maximum powers $P_{\max,lk}, \forall l,k$; Select initial values $p_{lk}(0) = P_{\max,lk}, \forall l,k$; Compute the total power consumption $P_{\mathrm{tot}}(0) = \sum_{l=1}^L \sum_{k=1}^K p_{lk}(0)$; Set initial value $n=1$ and tolerance $\epsilon$.
	\begin{itemize}
		\item[1.] User~$k$ in cell~$l$ computes the standard interference function $	{I}_{lk} \left(\mathbf{p} (n-1) \right)$ using \eqref{eq:Ilk}.
		\item[2.] If ${I}_{lk} \left(\mathbf{p} (n-1) \right) > P_{\max,lk}$, update $p_{lk}(n) = P_{\max,lk}$. Otherwise, update $ p_{lk}(n) = 	{I}_{lk} \left(\mathbf{p} (n-1) \right) $.
		\item[3.] Repeat Steps $1,2$ with other users, then compute the ratio \fontsize{9}{9}{$\gamma (n) =$ $| P_{\mathrm{tot}}(n) - P_{\mathrm{tot}}(n-1) | /  P_{\mathrm{tot}}(n-1)$}.
		\item[4.] If $\gamma_l (n) \leq \epsilon$ $\rightarrow$ Set $p_{lk}^{\ast} = p_{lk}(n),\forall l,k,$ and Stop. Otherwise, set $n= n+1$ and go to Step $1$.
	\end{itemize}
	\textbf{Output}: A fixed point $p_{lk}^{\ast}$, $\forall l,k$. \vspace*{-0.0cm}
\end{algorithm}
\subsection{Softly Removing Unsatisfied Users}
Instead of allowing potential unsatisfied users to spend full data power,  one can reduce their power with the goal to degrade mutual interference to the others. This policy might ameliorate the number of satisfied users in the entire network. The idea is in detail that: At first, every user improves the transmission quality by spending more power to each data symbol. This target can be achieved by, for example, simply constructing the standard inference functions as in the previous subsection. If at the limited power budget, the required SE cannot be achieved, unsatisfied users will reduce data power. We then mathematically suggest an update of the data powers along iterations as in Theorem~\ref{theorem:Standardfunction}.
\begin{theorem} \label{theorem:Standardfunction}
From the initial values $p_{lk}(0) = P_{\max,lk}, \forall l,k,$ if the data power of user~$k$ in cell~$l$ is updated at iteration~$n$ as
\begin{multline} \label{eq:UpdatedPowerv1}
p_{lk}(n) =  f_{lk} \left( \mathbf{p}(n-1) \right)  \\
= \begin{cases}
I_{lk} \left( \mathbf{p} (n-1) \right),& \mbox{if } I_{lk} \left( \mathbf{p} (n-1) \right) \leq P_{\max,lk}, \\
\frac{P_{\max,lk}^2}{I_{lk} \left(\mathbf{p} (n-1) \right)}, & \mbox{if } I_{lk} \left( \mathbf{p} (n-1) \right) > P_{\max,lk}, 
\end{cases}
\end{multline}
then the iterative approach converges to a fixed point.
\end{theorem}
\begin{proof}
The proof is first to confirm that the updated power policy in \eqref{eq:UpdatedPowerv1} follows a so-called two-sided  function and the convergence is then established. The detailed proof is available in Appendix~\ref{Appendix:Standardfunction}.
\end{proof} 
This theorem provides a procedure to minimize the total transmit power in the network and coping with the congestion issue based on the standard interference function defined for each user as in \eqref{eq:Ilk}. If $I_{lk}(\mathbf{p}(n-1))$ is less than the maximum power $P_{\max, lk}$ then the data power of user~$k$ in cell~$l$ is updated based on \eqref{eq:plkn}, same as what has done in Algorithm~\ref{Algorithm1}. The main distinction is to prevent any unsatisfied user from transmitting full power whenever the congestion issue happens, i.e. $I_{lk} (\mathbf{p}(n-1)) > P_{\max,lk}$. In particular, the data power of a unsatisfied user scales down with the total mutual interference and noise level, which contains in $I_{lk}(\mathbf{p}(n-1))$. By doing this, the mutual interference from this unsatisfied user to the others should be reduced, and hence there is chance for the remaining users to get their required SEs. The proposed optimization approach is summarized in Algorithm~\ref{Algorithm2}. The per iteration complexity is $\mathcal{O}\left(L^2 K^2 + 3 |\mathcal{P}|LK \right)$, thus the computational complexity of Algorithm~\ref{Algorithm2} is in the order of $\mathcal{O}\left(N_s L^2 K^2 + 3 N_s |\mathcal{P}|LK \right)$, where $N_s$ is the number of iterations needed for this algorithm converges. Furthermore, Theorem~\ref{theorem:Standardfunction} analytically proves the convergence to a fixed point, whose property is stated in Remark~\ref{Remarlk:Property}.
\begin{remark}\label{Remarlk:Property}
The proposed algorithms enable to work in both feasible and infeasible domain such that a fixed point to problem~\eqref{Prob:TotalPowerOptv3} can be obtained. For realizations of user locations that result in feasible domains, the fixed point obtained by those algorithms is unique, which is the global optimum. The main difference between the two algorithms is at the policy to assign data powers whenever the congestion issue appears. While Algorithm~\ref{Algorithm1} allocates the maximum data power to users when their SINR constraints are not satisfied, Algorithm~\ref{Algorithm2} reduces the data power. As a consequence, for an infeasible domain to problem~\eqref{Prob:TotalPowerOptv3}, the fixed point obtained by each algorithm may be different from each other.  

We notice that it is straightforward to extend the proposed algorithms to the total downlink energy consumption optimization problem with the per-user power constraints. The extension is not trivial if one considers the per-BS total limited power budgets and a primal-dual decomposition approach might be utilized to allocate the downlink power coefficients based on the standard interference functions.
\end{remark}
\begin{algorithm}[t]
	\caption{Data power allocation to problem~\eqref{Prob:TotalPowerOptv3} by softly removing unsatisfied users} \label{Algorithm2}
	\textbf{Input}:  Define maximum powers $P_{\max,lk}, \forall l,k$; Select initial values $p_{lk}(0) = P_{\max,lk}, \forall l,k$; Compute the total power consumption $P_{\mathrm{tot}}(0) = \sum_{l=1}^L \sum_{k=1}^K p_{lk}(0)$; Set initial value $n=1$ and tolerance $\epsilon$.
	\begin{itemize}
		\item[1.] User~$k$ in cell~$l$ computes the standard interference function $	{I}_{lk} \left(\mathbf{p} (n-1) \right)$ using \eqref{eq:Ilk}.
		\item[2.] If ${I}_{lk} \left(\mathbf{p} (n) \right) < P_{\max,lk}$, update $ p_{lk}(n) = {I}_{lk} \left(\mathbf{p} (n-1) \right) $. Otherwise, update $p_{lk}(n) =  P_{\max,lk}^2 / {I}_{lk} \left(\mathbf{p} (n-1) \right)$.
		\item[3.] Repeat Steps $1,2$ with other users, then compute the ratio \fontsize{9}{9}{$\gamma (n) =$ $| P_{\mathrm{tot}}(n) - P_{\mathrm{tot}} (n-1) | /   P_{\mathrm{tot}}(n-1)$}.
		\item[4.] If $\gamma_l (n) \leq \epsilon$ $\rightarrow$ Set $p_{lk}^{\ast} = p_{lk}(n),\forall l,k,$ and Stop. Otherwise, set $n= n+1$ and go to Step $1$.
	\end{itemize}
	\textbf{Output}: A fixed point $ p_{lk}^{\ast}$, $\forall l,k$. 
\end{algorithm}
\begin{figure*}[t]
	\begin{minipage}{0.48\textwidth}
		\centering
		\includegraphics[trim=0.8cm 0cm 1.4cm 0.5cm, clip=true, width=3.0in]{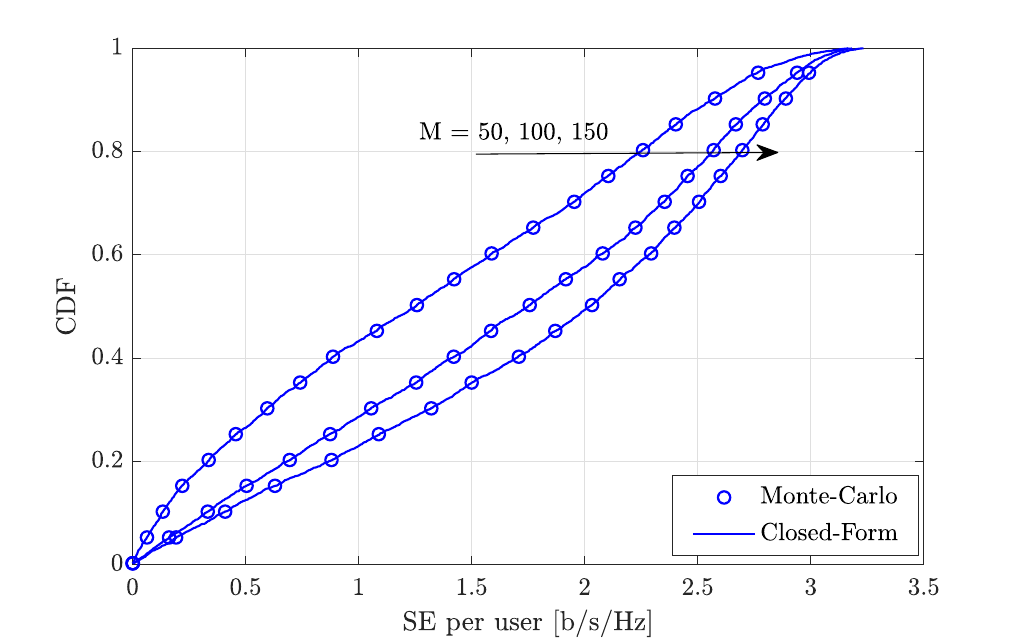} \vspace*{-0.25cm}
		\caption{The CDF of the uplink SE per user [b/s/Hz] with Monte-Carlo simulation and closed-form expression with $S_{lk}^{l'} = 21, \forall l,l',k$. }
		\label{FigMonteCarloClosedForm}
		\vspace*{-0.2cm}
	\end{minipage}
	\hfill
	\begin{minipage}{0.48\textwidth}
		\centering
		\includegraphics[trim=0.8cm 0cm 1.4cm 0.5cm, clip=true, width=3.0in]{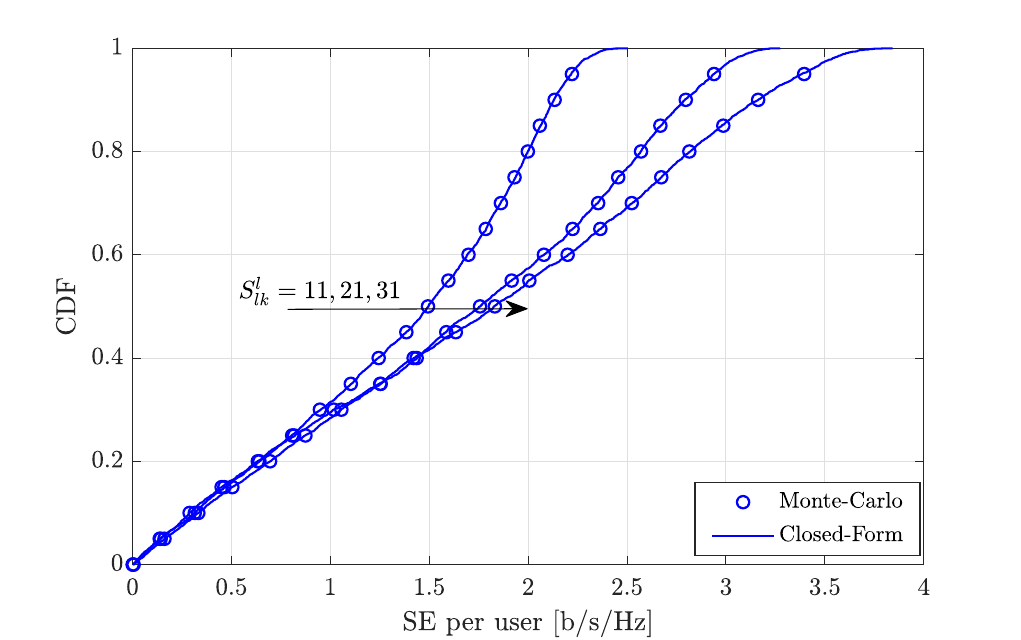} \vspace*{-0.25cm}
		\caption{The CDF of the uplink SE per user [b/s/Hz] with Monte-Carlo simulation and closed-form expression with $M=100$.}
		\label{FigMonteCarloClosedFormv1}
		\vspace*{-0.2cm}
	\end{minipage}
	\vspace*{-0.2cm}
\end{figure*}
\section{Numerical Results} \label{Sec:NumericalResults}
We consider a Massive MIMO system with $L=4$ square cells in a $1$~km$^2$ area, each serving $K=5$ users. All the users are uniformly distributed within its cell with the distance to the BS no less than $35$~m. Each coherence book has $\tau_c = 200$ symbols and there are $\tau_p = 5$ orthogonal pilot signals with the power $\hat{p}_{lk} = P_{\max,lk}= 200$~mW, $\forall l,k$. Without the loss of generality, the users with same index in all cells sharing a orthogonal pilot signal. The system bandwidth is $20$~MHz and the noise variance is $-96$~dBm with the noise figure $5$~dB. The large-scale fading coefficient [dB] of user~$k$ in cell~$l$ and BS~$l'$ is modeled based on the 3GPP LTE specifications \cite{LTE2010a} as
\begin{equation} \label{eq:LargeScale}
\beta_{lk}^{l'} = -128.1 - 37.6 \log_{10} \left( d_{lk}^{l'} / 1 \mbox{km}\right) + z_{lk}^{l'},
\end{equation}
where $d_{lk}^{l'} > 35$~m is the distance between user~$k$ in cell~$l$ and BS~$l'$; $z_{lk}^{l'}$ is the shadow fading coefficient, which follows a Gaussian distribution with zero mean and standard deviation $7$~dB. The covariance matrices are computed by using \cite[(13) and (16)]{van2016multi}. In the proposed algorithms (Algorithms~\ref{Algorithm1} and \ref{Algorithm2}), we set $\epsilon = 0.001$, except Fig.~\ref{FigConverg} which visualizes the convergence property. For feasible systems, the global optimum obtained by utilizing interior point methods from previous works like \cite{senel2019joint,van2020power} are included for comparison.\footnote{In \cite{senel2019joint,van2020power}, user locations and shadow fading realizations resulting in a feasible domain have been considered for conveniences to utilize the interior-point methods. If only one user is not satisfied with its SE requirement, it is sufficient to create an infeasible set. Consequently, the problem lacks a feasible solution.}

Figure~\ref{FigMonteCarloClosedForm} shows the cumulative distribution function (CDF) of SE per user [b/s/Hz] to verify the correctness of the closed-form expression of the uplink SE for each user obtained in Theorem~\ref{Theorem:ClosedFormMR}. There are $21$ scatterers per communication link and all users spend full power for the data transmission. Particularly, the closed-form expression result matches very well Monte-Carlo simulation result for all the considered number of BS antennas. This figure also illustrates the SE per user getting better when each BS is equipped with more antennas. Each user can be served by a data rate increasing from $1.3$ [b/s/Hz] to $1.8$ [b/s/Hz] on average if the number of BS antennas increases from $50$ to $150$, which is a $38.5\%$ data rate improvement. From this amount of antennas added, the median SE gets significantly better with a $60\%$ data rate improvement as a consequence of the SE per user increasing from $1.25$ [b/s/Hz] to $2$ [b/s/Hz].

Figure~\ref{FigMonteCarloClosedFormv1} plots the CDF of SE per user [b/s/Hz] with a different number of scatterers. Each BS is equipped with $100$ antennas. All the Monte-Carlo simulations producing the same SE as the closed-form expression verifies the correctness of Theorem~\ref{Theorem:ClosedFormMR} when the number of scatterers varies. Clearly, the SE per user gets better for rich scattering environments. On average, a notable gain of $1.25\times$  in SE is obtained if each channel has $21$ scatterers instead of $11$ scatterers. However, the SE has a small gai, e.g., with only $6.6\%$ if the propagation environment has $31$ scatterers. Therefore, Fig.~\ref{FigMonteCarloClosedFormv1} unveils a slow growth of the SE as a function of the scatterer number. At $95\%$-likely, the three considered scenarios provide the same SE with $0.16$ [b/s/Hz] without data power control. Consequently, it seems that poor scattering environments affect  the worst SE slightly.

\begin{figure*}[t]
	\begin{minipage}{0.48\textwidth}
		\centering
		\includegraphics[trim=0.8cm 0cm 1.4cm 0.5cm, clip=true, width=3.0in]{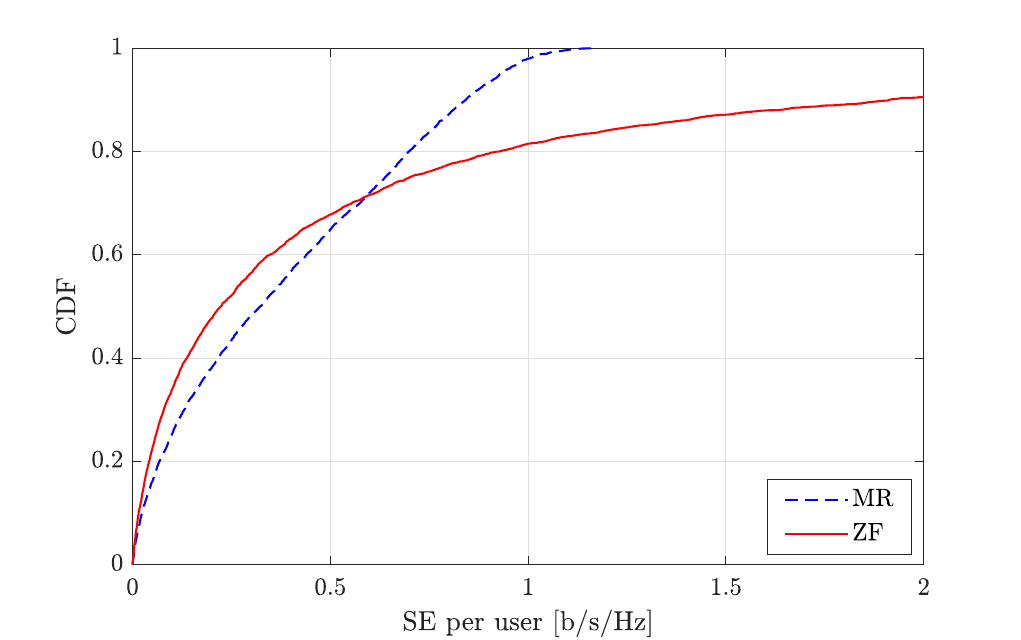} \vspace*{-0.25cm}
		\caption{The CDF of SE per user with the different linear combining techniques, $M= 100$, and $S_{lk}^{l'} = 3, \forall l,l',k$.}
		\label{FigDifferentCombining}
		\vspace*{-0.2cm}
	\end{minipage}
	\hfill
	\begin{minipage}{0.48\textwidth}
		\centering
		\includegraphics[trim=0.8cm 0cm 1.4cm 0.5cm, clip=true, width=3.0in]{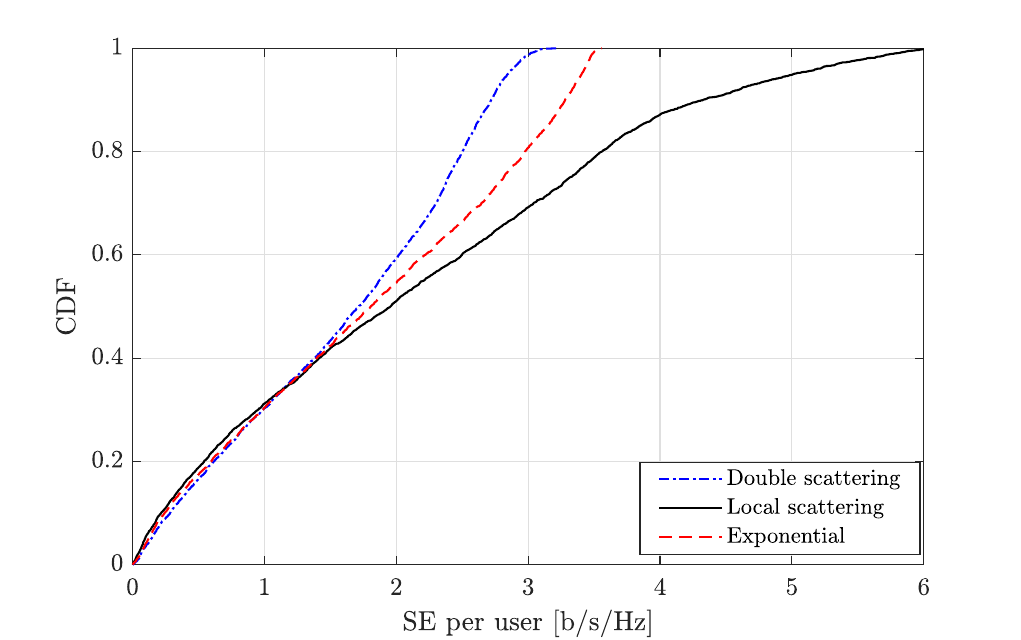} \vspace*{-0.25cm}
		\caption{The CDF of SE per user with the different spatially correlated models, $M=100$.}
		\label{FigDifferentCovarianceModels}
		\vspace*{-0.2cm}
	\end{minipage}
	\vspace*{-0.2cm}
\end{figure*}

Figure~\ref{FigDifferentCombining} shows the CDF of SE per user  [b/s/Hz] for a system with either MR or ZF combining technique with a small number of scatterers per each propagation channel. The transmit power per symbol is $50$~mW and the large-scale fading coefficients are computed similar to \eqref{eq:LargeScale} but with the penetration loss of $20$~dB. ZF generally provides better performance than MR since it cancels out mutual interference more effectively \cite{van2016multi}. On average, a system with MR combining is still the baseline that offers less than that of utilizing ZF combining. Nonetheless, Fig.~\ref{FigDifferentCombining} demonstrates the sensitivity of ZF when the propagation environment lacks scatterers in many user locations and shadow fading realizations which result in low-rank channels. Consequently,  MR outperforms ZF about $45.5\%$ at the median SE.

Figure~\ref{FigDifferentCovarianceModels} presents the CDF of SE per user by utilizing the different spatial correlation channel models. There are $21$ scatterers for each propagation link  with the double scattering channel model. The exponential correlation model is defined as in \cite{Chien2018a} with the correlation magnitude $0.9$, while the local scattering channel model is defined in \cite{massivemimobook} with $6$ scattering clusters, the angular standard deviation $5^\circ$, and the antenna spacing of the half wavelength. By assuming that the scattering clusters are in the  half-space in front of the BSs, the local scattering channel model offers the highest SE per user with up to $2.1$ [b/s/Hz] on average. The exponential correlation model provides the SE of about $1.8$ [b/s/Hz] per user. Meanwhile, the double scattering model yields to the lowest SE with only $1.6$ [b/s/Hz] due to taking both the local scattering property and rank deficiency into account.

\begin{figure*}[t]
	\begin{minipage}{0.48\textwidth}
		\centering
		\includegraphics[trim=0.6cm 0cm 1.4cm 0.5cm, clip=true, width=3.0in]{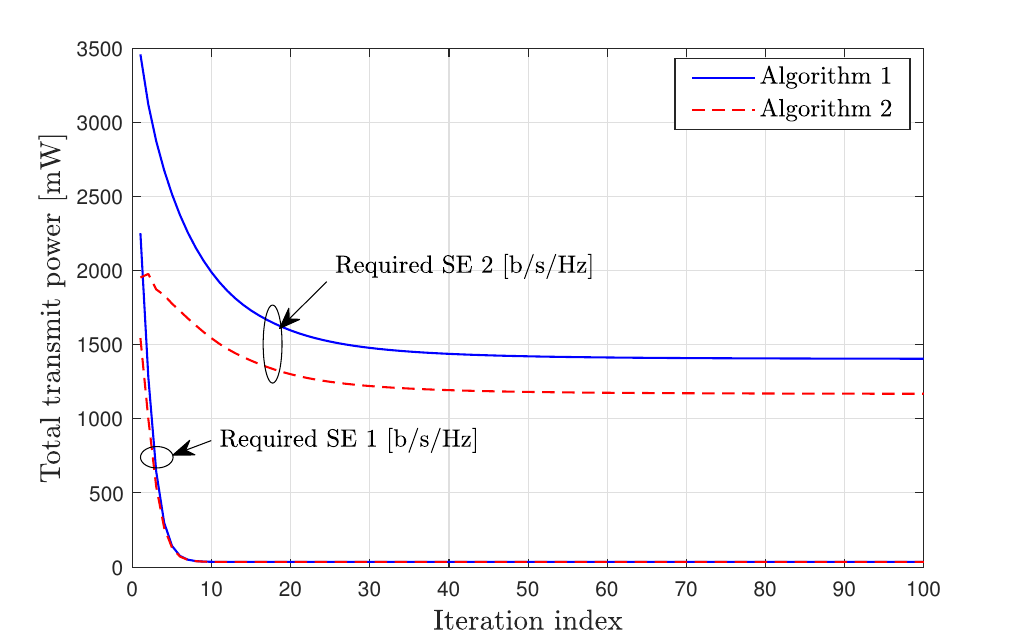} \vspace*{-0.25cm}
		\caption{The convergence of Algorithms~\ref{Algorithm1} and \ref{Algorithm2} with the different required SEs at the users, $M= 100$, and $S_{lk}^{l'} = 21, \forall l,l',k$. }
		\label{FigConverg}
		\vspace*{-0.2cm}
	\end{minipage}
	\hfill
	\begin{minipage}{0.48\textwidth}
		\centering
		\includegraphics[trim=0.8cm 0cm 1.4cm 0.5cm, clip=true, width=3.0in]{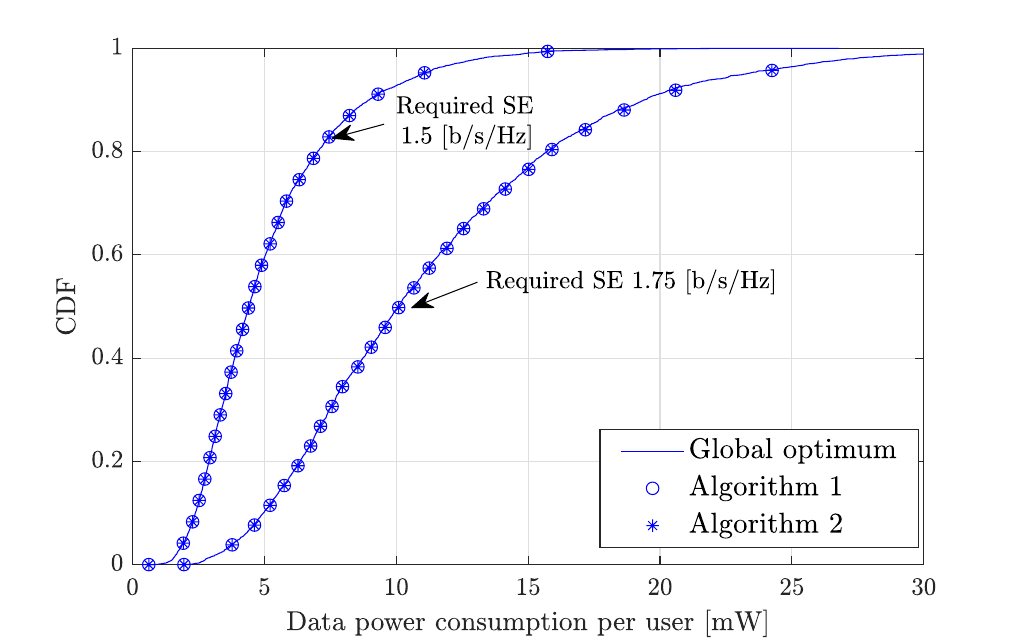} \vspace*{-0.25cm}
		\caption{The CDF of the power consumption per user [mW] for feasible systems with the different required SEs at the users, $M=100$, and $S_{lk}^{l'} = 21, \forall l,k,l'$. }
		\label{FigPowerFeas}
		\vspace*{-0.2cm}
	\end{minipage}
	\vspace*{-0.2cm}
\end{figure*}
Figure~\ref{FigConverg} illustrates the convergence of Algorithms~\ref{Algorithm1} and \ref{Algorithm2} by utilizing two different required SEs. They converge fast to a fixed point after a few tens of iterations. If each user requests a SE $1$ [b/s/Hz], the proposed algorithms need less than $10$ iterations to reach convergence, which is the same fixed point. This fixed point is the global optimum since the optimization problem is always feasible for the user locations and shadow fading realizations have been generated. When the required SEs expand to $2$ [b/s/Hz], the proposed algorithms require around $40$ iterations to approach the optimum. The convergence rate is therefore slower when the SE requirements enlarge. This SE setting also manifests the benefits of Algorithm~\ref{Algorithm2}, which yields $20\%$ less the total transmit power than Algorithm~\ref{Algorithm1}. On the other hand, the fixed point obtained by each algorithm is different from each other. 

We show the CDF of the data power consumption [mW] consumed by each user in Fig.~\ref{FigPowerFeas} for feasible systems with the two different required SEs. Matched well with the claim in Remark~\ref{Remarlk:Property} for feasible systems, the proposed algorithms provide a unique fixed point that is the global optimum as what has obtained by the interior-point methods. Additionally, data power escalates when users require higher SEs. With the required SE $1.5$ [b/s/Hz], each user only spends $5.2$ mW for each data symbol on average. However, it drastically grows to $11.4$ mW (corresponding to $2.2 \times$ more power) with the required SE $1.75$ [b/s/Hz]. Both the considered SE settings illustrate significant reductions of transmit power compared to the scenario dedicating full power to the data symbols. Particularly, all the users consume $38.5 \times$ and $17.5 \times$ less power than the full power transmission with the two considered SEs, respectively.

\begin{figure*}[t]
	\begin{minipage}{0.48\textwidth}
		\centering
		\includegraphics[trim=0.8cm 0cm 1.4cm 0.5cm, clip=true, width=3.0in]{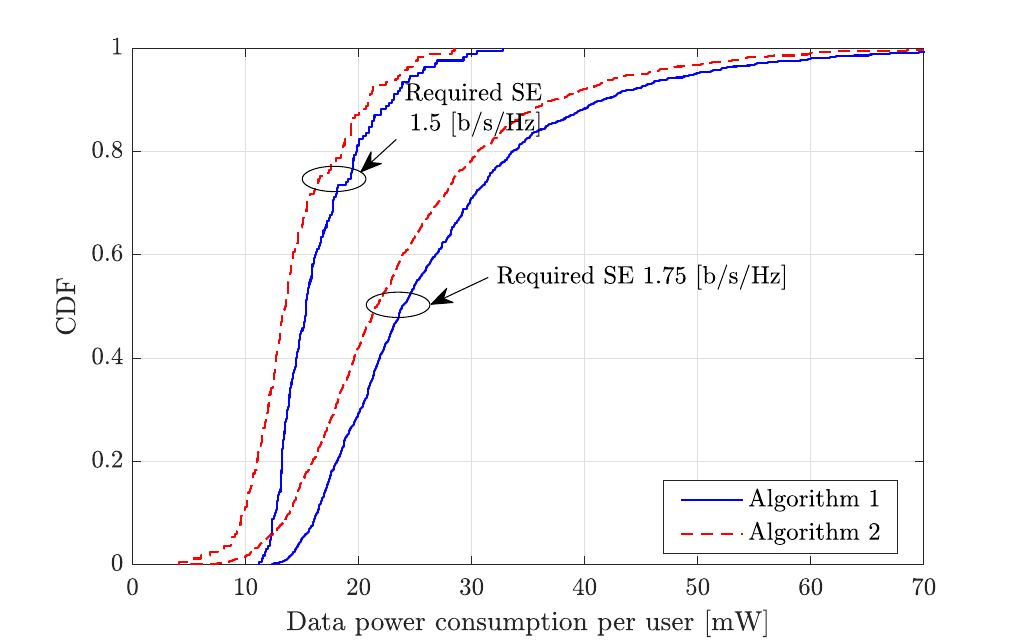} \vspace*{-0.25cm}
		\caption{The CDF of the power consumption per user [mW] for infeasible systems with the different required SEs at the users, $M=100$, and $S_{lk}^{l'} = 21, \forall l,k,l'$.}
		\label{FigPowerInfeas}
		\vspace*{-0.2cm}
	\end{minipage}
	\hfill
	\begin{minipage}{0.48\textwidth}
		\centering
		\includegraphics[trim=0.8cm 0cm 1.4cm 0.5cm, clip=true, width=3.0in]{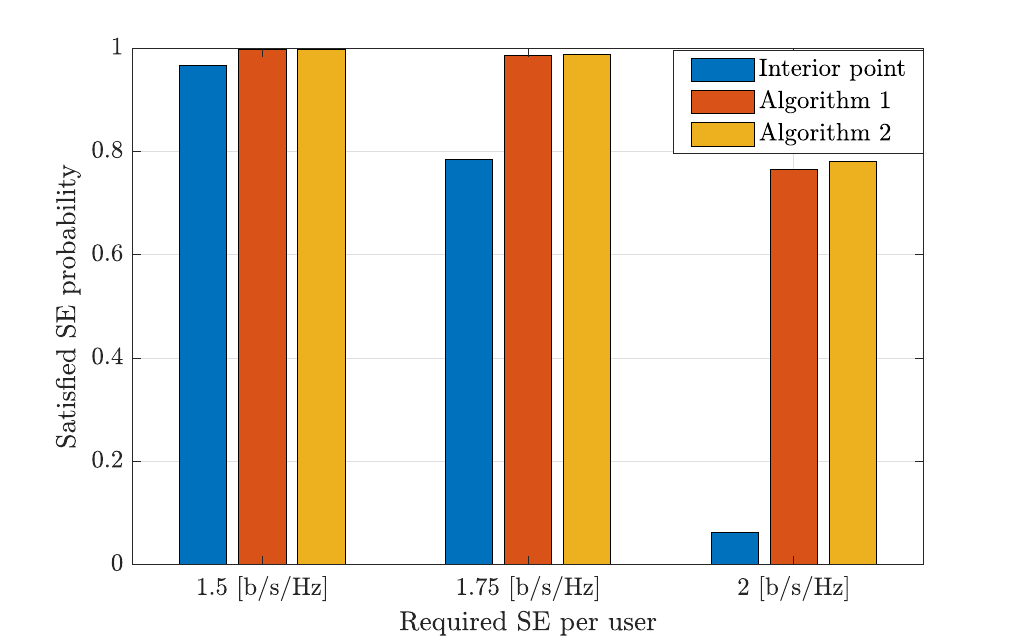} \vspace*{-.25cm}
		\caption{The satisfied SE probability versus the different required SE per user for a system with $M=100$ and $S_{lk}^{l'} = 21, \forall l,k,l'$.}
		\label{FigProb}
		\vspace*{-0.2cm}
	\end{minipage}
	\vspace*{-0.2cm}
\end{figure*}
Figure~\ref{FigPowerInfeas} displays the  CDF of the data power consumption [mW] per user for infeasible systems. It is the main interest of this paper when working with multiple access in Massive MIMO communications since there is no global optimum to obtain or compare against. All the users consume non-zero powers at the fixed points identified Algorithms~\ref{Algorithm1} and \ref{Algorithm2}. The trend that more data power is needed when the users require higher SEs has still remained. In more detail, the data power obtained by Algorithm~\ref{Algorithm1} grows $1.6 \times$  from $16.6$~mW to $27.0$~mW when the required SE increases from $1.5$ [b/s/Hz] to $1.75$ [b/s/Hz]. The data power increases $1.7\times$ from $14.5$~mW to $24.1$~mW if Algorithm~\ref{Algorithm2} is exploited. Moreover, the data power consumption per user obtained by Algorithm~\ref{Algorithm1} is $12.3\%$ and $15.1\%$ higher than by Algorithm~\ref{Algorithm2}.

Figure~\ref{FigProb} plots the satisfied SE probability defined as the fraction of random user locations and shadow fading realizations in which the users can be served by the required SEs. If each user requires an SE $1.5$ [b/s/Hz], all the benchmarks provide an overwhelming satisfied SE probability. For instance, the interior-point methods offer $96.7\%$ user locations and shadow fading realizations with the required SEs. Meanwhile, the proposed algorithms offer a satisfied SE probability $99.8\%$. However, the interior-point methods will perform worse with higher SE requirements since only one user is sufficient to create an empty feasible set as aforementioned in Section~\ref{SubSec:FeasInfeas}, especially only $6.3\%$ users satisfied the required SE $2$ [b/s/Hz]. In contrast, the proposed algorithms still offer a satisfied SE probability of more than $75\%$. Furthermore, Algorithm~\ref{Algorithm2} slightly performs better than Algorithm~\ref{Algorithm1} in those required SE settings.

\begin{figure*}[t]
	\begin{minipage}{0.48\textwidth}
		\centering
		\includegraphics[trim=0.8cm 0cm 1.4cm 0.5cm, clip=true, width=3.0in]{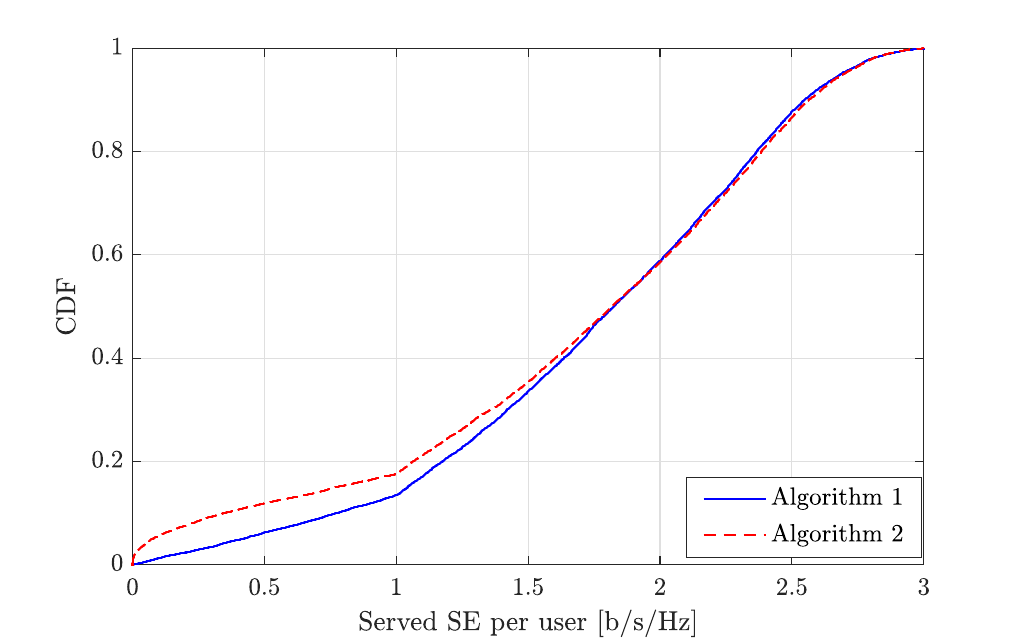} \vspace*{-.25cm}
		\caption{The CDF of served SE per user [b/s/Hz] with $M=100$, $S_{lk}^{l'} = 21, \forall l,k,l'$, and the required SEs uniformly varying in the range [1, 3] [b/s/Hz].}
		\label{FigCDFPowerVariedQoS}
		\vspace*{-0.2cm}
	\end{minipage}
	\hfill
	\begin{minipage}{0.48\textwidth}
		\centering
		\includegraphics[trim=0.8cm 0cm 1.4cm 0.5cm, clip=true, width=3.0in]{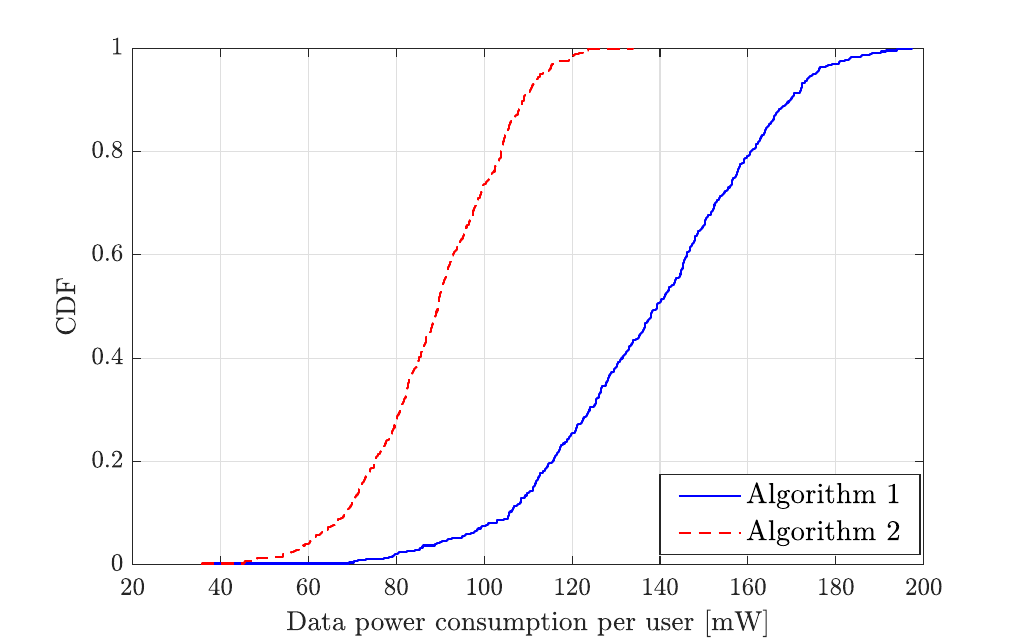} \vspace*{-0.25cm}
		\caption{The CDF of data power consumption [mW] with $M=100$, $S_{lk}^{l'} = 21, \forall l,k,l'$, and the required SEs uniformly varying in the range [1, 3] [b/s/Hz].}
		\label{FigCDFVariedQoS}
		\vspace*{-0.2cm}
	\end{minipage}
	\vspace*{-0.2cm}
\end{figure*}
Figure~\ref{FigCDFPowerVariedQoS} provides the served SE per user [b/s/Hz] when the users have different required SEs, which are uniformly distributed in the range $[1, 3]$ [b/s/Hz] over many user locations and shadowing fading realizations. The interior-point methods are not included since the optimization problem always has an empty feasible domain in this complicated scenario. Interestingly, Algorithm~\ref{Algorithm1} performs pretty better than Algorithm~\ref{Algorithm2}  since the former gives $86.5\%$ users satisfied their SEs, while the latter is only $82.5\%$. However,  Fig.~\ref{FigCDFVariedQoS} indicates that Algorithm~\ref{Algorithm2} produces a fixed point that has much lower power consumption than Algorithm~\ref{Algorithm1}. The saving power is about $54.7\%$ on average thanks to the data reduction policy in \eqref{eq:UpdatedPowerv1} whenever the congestion issue appears.
\begin{figure}[t]
	\centering
	\includegraphics[trim=0.8cm 0cm 1.4cm 0.5cm, clip=true, width=3.0in]{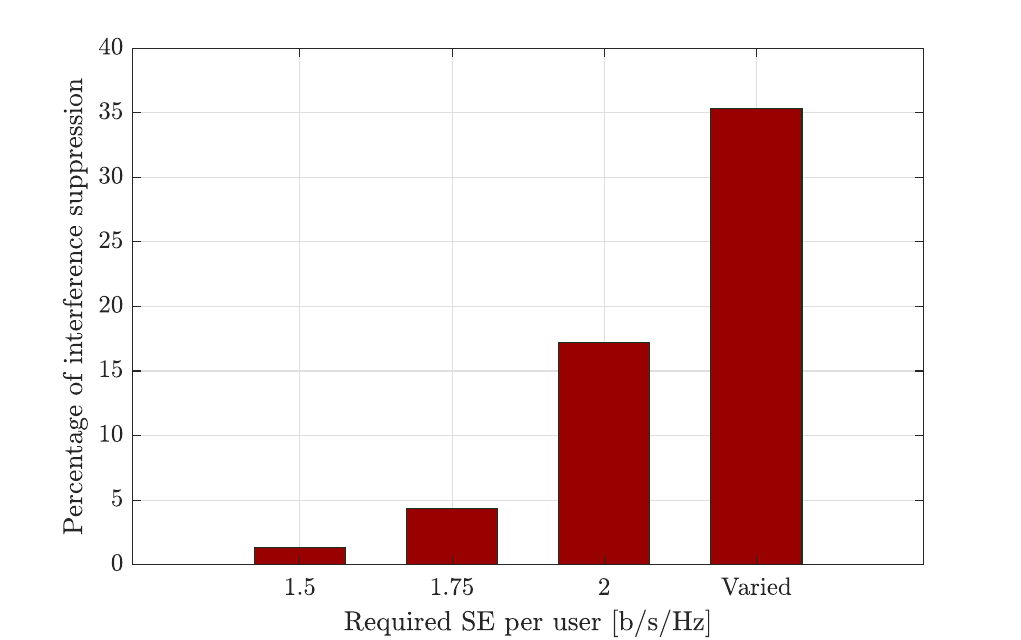} \vspace*{-0.25cm}
	\caption{The interference suppression obtained by Algorithm~\ref{Algorithm2} compared to Algorithm~\ref{Algorithm1} as a function of the required SE per user with $M=100$ and $S_{lk}^{l'}= 21, \forall l,k,l'$.}
	\label{FigPercentInterference}
	\vspace*{-0.2cm}
\end{figure}

Figure~\ref{FigPercentInterference} shows the percentage of interference suppression obtained by Algorithm~\ref{Algorithm2} in a comparison to Algorithm~\ref{Algorithm1} by utilizing the different required SEs per user. Softly removing unsatisfied users generates less mutual interference than  spending the maximum transmit power on those users, especially when the SE requirements are high. For instance, mutual interference from Algorithm~\ref{Algorithm2} is only $1.3\%$ less than that of Algorithm~\ref{Algorithm1} if the required SE per user is $1.5$ [b/s/Hz]. However, the mutual interference suppression gains up to $17.2\%$ with the SE requirement $2$ [b/s/Hz]. In particular, Algorithm~\ref{Algorithm2} suppresses mutual interference significantly when each user has its own SE requirement varied in the range from $1$ [b/s/Hz] to $3$ [b/s/Hz] with the mutual interference suppression of about $35.4\%$. We therefore conclude the effectiveness of the second algorithm compared with the first one.
\section{Conclusion} \label{Sec:Conclusion}
This paper has analyzed the system performance of Massive MIMO systems with an arbitrary number of BS antennas, users, and scatterers by utilizing the double scattering channel model, rather than the asymptotic regime as in previous works. The closed-form expression of the uplink SE per user was first computed, then the asymptotic performance was obtained. We further formulated and solved a total transmit power minimization problem with the required SE constraints and limited power budget. We proposed two algorithms to handle effectively the congestion issue that often happens since multiple users are simultaneously connecting to the network and sharing the same time and frequency resources. The solutions to those algorithms are quite similar to each other if the  required SEs can be almost satisfied with the given power budget. In contrast, Algorithm~\ref{Algorithm2} outperforms Algorithm~\ref{Algorithm1} in phenomena where the SE requirements are vastly different and many users cannot be served with the required SEs. 
\appendix
\subsection{Proof of Lemma~\ref{corollary:ChannelProperties}} \label{Appendix:ChannelProperties}
For a given matrix $\mathbf{B}$, we first compute the statistical information of the two channels $\mathbf{h}_{l'k'}^{l}$ and $ \mathbf{h}_{l''k''}^l$ when $(l',k') \neq (l'',k'')$ by averaging over the different realizations of small-fading coefficients as
\begin{multline} \label{eq:Property1}
	\mathbb{E} \left\{ \Big| \big(\mathbf{h}_{l'k'}^{l}\big)^{\rm H} \mathbf{B}  \mathbf{h}_{l''k''}^l \Big|^2 \right\} =  \\ 
	\mathrm{tr} \left( \mathbf{B} \mathbb{E} \Big\{ \mathbf{h}_{l'' k''}^l \big( \mathbf{h}_{l'' k''}^l\big)^{\rm H} \Big\} \mathbf{B}^{\rm H}  \mathbb{E} \Big\{ \mathbf{h}_{l' k'}^l \big( \mathbf{h}_{l' k'}^l\big)^{\rm H} \Big\} \right).
\end{multline}
The first expectation in the right-hand side of \eqref{eq:Property1} is computed by plugging the definition of the double-scattering channel model in \eqref{eq:Channel} as
\begin{equation} \label{eq:Expv1}
	\begin{split}
		&\mathbb{E} \Big\{ \mathbf{h}_{l'' k''}^l \big( \mathbf{h}_{l'' k''}^l\big)^{\rm H} \Big\} = \frac{\beta_{l''k''}^l }{S_{l'' k''}^l} \mathbb{E} \Big\{ \big(\mathbf{R}_{l''k''}^{l}\big)^{1/2} \mathbf{G}_{l''k''}^{l} \big(\widetilde{\mathbf{R}}_{l''k''}^{l} \big)^{1/2}   \times \\
		& \mathbb{E} \Big\{ \mathbf{g}_{l''k''}^{l} \big( \mathbf{g}_{l''k''}^{l}\big)^{\rm H} \Big\}  \big(\widetilde{\mathbf{R}}_{l''k''}^{l} \big)^{1/2} \big( \mathbf{G}_{l''k''}^{l} \big)^{\rm H} \big(\mathbf{R}_{l''k''}^{l}\big)^{1/2} \Big\}  \\
		&= \frac{\beta_{l''k''}^l }{S_{l'' k''}^l}     \big(\mathbf{R}_{l''k''}^{l}\big)^{1/2} \mathbb{E} \Big\{\mathbf{G}_{l''k''}^{l} \widetilde{\mathbf{R}}_{l''k''}^{l}  \big( \mathbf{G}_{l''k''}^{l} \big)^{\rm H} \Big\} \big(\mathbf{R}_{l''k''}^{l}\big)^{1/2} \\
		& = \beta_{l''k''}^l d_{l''k''}^l \mathbf{R}_{l''k''}^{l},
	\end{split}
\end{equation}
where the last equality of \eqref{eq:Expv1} is obtained by utilizing \cite[Lemma~$8$]{van2020massivechannels} to compute the covariance matrix of the circularly symmetric complex Gaussian matrix $\mathbf{G}_{l''k''}^l$ for a given deterministic matrix $\mathbf{R}_{l''k''}^{l}$. Following a similar manner, the second expectation in the right-hand side of \eqref{eq:Property1} is computed in closed form as
\begin{equation} \label{eq:Expv2}
	\mathbb{E} \Big\{ \mathbf{h}_{l' k'}^l \big( \mathbf{h}_{l' k'}^l\big)^{\rm H} \Big\} = \beta_{l'k'}^l d_{l'k'}^l \mathbf{R}_{l'k'}^{l}.
\end{equation}
Plugging \eqref{eq:Expv1} and \eqref{eq:Expv1} into \eqref{eq:Property1}, we obtain the result as shown in \eqref{eq:Cor1eq9}. For a given deterministic matrix $\mathbf{B}$, the statistical information of the channel $\mathbf{h}_{l'k'}^{l}$ is computed as
\fontsize{10}{10}{\begin{equation} \label{eq:Expv3}
		\begin{split}
			&\mathbb{E} \left\{ \Big| \big(\mathbf{h}_{l'k'}^l \big)^{\rm H} \mathbf{B}  \mathbf{h}_{l'k'}^l \Big|^2 \right\}  = \frac{ \big(\beta_{l'k'}^l \big)^2 }{ \big(S_{l' k'}^l\big)^2} \mathbb{E} \Big\{ \Big|  \big( \mathbf{g}_{l'k'}^{l}\big)^{\rm H}  \big(\widetilde{\mathbf{R}}_{l'k'}^{l} \big)^{1/2} \big( \mathbf{G}_{l'k'}^{l} \big)^{\rm H} \\
			& \times \big(\mathbf{R}_{l'k'}^{l}\big)^{1/2} \mathbf{B}  \big(\mathbf{R}_{l'k'}^{l}\big)^{1/2} \mathbf{G}_{l'k'}^{l} \big(\widetilde{\mathbf{R}}_{l'k'}^{l} \big)^{1/2}  \mathbf{g}_{l'k'}^{l}  \Big|^2 \Big\}\\
			& = \frac{\big(\beta_{l'k'}^l \big)^2}{ \big( S_{l' k'}^l \big)^2} \mathbb{E} \left\{ \Big\| \big(\widetilde{\mathbf{R}}_{l'k'}^{l} \big)^{1/2}  \mathbf{g}_{l'k'}^{l} \Big\|^4 \left| \frac{ \big( \mathbf{g}_{l'k'}^{l}\big)^{\rm H} \big(\widetilde{\mathbf{R}}_{l'k'}^{l} \big)^{1/2} \big( \mathbf{G}_{l'k'}^{l} \big)^{\rm H}}{ \Big\| \big(\widetilde{\mathbf{R}}_{l'k'}^{l} \big)^{1/2}  \mathbf{g}_{l'k'}^{l} \Big\|}  \right. \right. \\
			& \left. \left.  \times \big(\mathbf{R}_{l'k'}^{l}\big)^{1/2} \mathbf{B}  \big(\mathbf{R}_{l'k'}^{l}\big)^{1/2} \frac{\mathbf{G}_{l'k'}^{l} \big(\widetilde{\mathbf{R}}_{l'k'}^{l} \big)^{1/2}  \mathbf{g}_{l'k'}^{l}}{ \Big\| \big(\widetilde{\mathbf{R}}_{l'k'}^{l} \big)^{1/2}  \mathbf{g}_{l'k'}^{l} \Big\|}  \right|^2 \right\},
		\end{split}
\end{equation}}
where the last equality of \eqref{eq:Expv3} is obtained by utilizing the normalization term $\Big\| \big(\widetilde{\mathbf{R}}_{l'k'}^{l} \big)^{1/2}  \mathbf{g}_{l'k'}^{l} \Big\|$. Let us introduce the new optimization variable $\mathbf{z}_{l'k'}^l$, which is defined as
\begin{equation}
	\mathbf{z}_{l'k'}^l =  \frac{\mathbf{G}_{l'k'}^{l} \big(\widetilde{\mathbf{R}}_{l'k'}^{l} \big)^{1/2}  \mathbf{g}_{l'k'}^{l}}{ \Big\| \big(\widetilde{\mathbf{R}}_{l'k'}^{l} \big)^{1/2}  \mathbf{g}_{l'k'}^{l} \Big\|},
\end{equation}
then it is straightforward to prove that $\mathbf{z}_{l'k'}^l \sim \mathcal{CN}\left( \mathbf{0}, \mathbf{I}_M \right)$, and is independent of $\mathbf{g}_{l'k'}^l$. Thus, \eqref{eq:Expv3} is equivalent to the following expression 
\begin{equation} \label{eq:2ndProperty}
	\begin{split}
		&\mathbb{E} \left\{ \Big| \big(\mathbf{h}_{l'k'}^l \big)^{\rm H} \mathbf{B}  \mathbf{h}_{l'k'}^l \Big|^2 \right\} = \frac{\big(\beta_{l'k'}^l\big)^2 }{\big(S_{l' k'}^l\big)^2} \mathbb{E} \left\{ \Big\| \big(\widetilde{\mathbf{R}}_{l'k'}^{l} \big)^{1/2}  \mathbf{g}_{l'k'}^{l} \Big\|^4 \right\} \times \\
		& \quad \mathbb{E} \left\{ \Big| \big(\mathbf{z}_{l'k'}^l\big)^{\rm H} \big(\mathbf{R}_{l'k'}^{l}\big)^{1/2} \mathbf{B}  \big(\mathbf{R}_{l'k'}^{l}\big)^{1/2} \mathbf{z}_{l'k'}^l  \Big|^2 \right\}\\
		&= \frac{ \big(\beta_{l'k'}^l\big)^2 }{\big(S_{l' k'}^l\big)^2}\left(\Big| \mathrm{tr} \big( \widetilde{\mathbf{R}}_{l'k'}^l \big) \Big|^2 +    \mathrm{tr} \Big( \big( \widetilde{\mathbf{R}}_{l'k'}^l \big)^2 \Big) \right) \times \\
		& \quad \left( \Big|\mathrm{tr} \big( \mathbf{R}_{l'k'}^l \mathbf{B} \big) \Big|^2 + \mathrm{tr} \big( \mathbf{R}_{l'k'}^l \mathbf{B} \mathbf{R}_{l'k'}^l \mathbf{B}^H \big) \right),
	\end{split}
\end{equation}
where the last equality in \eqref{eq:2ndProperty} is obtained by utilizing \cite[Lemma~9]{van2020massivechannels} to compute the forth moment of zero-mean complex Gaussian variables, and then the result is obtained as in \eqref{eq:4moments} after doing some algebra.
\subsection{Proof of Lemma~\ref{lemma:ChannelEstPhaseAware}} \label{appendix:ChannelEstPhaseAware}
Following the similar approach as \cite[Lemma~3]{van2020massivechannels}, we can compute the correlation matrix of two channel vectors $\mathbf{h}_{lk}^{l'}$ and $\mathbf{h}_{l''k''}^{l'}$ by averaging over the different realizations of small-scale fading coefficients as
\begin{align} \label{eq:CovMatrix}
\mathbb{E} \Big\{ \mathbf{h}_{lk}^{l'} \big(\mathbf{h}_{l"k"}^{l'}\big)^{\rm H} \Big\} =  \begin{cases} \beta_{lk}^{l'} d_{lk}^{l'} \mathbf{R}_{lk}^{l'}, & \mbox{if }  (l,k) = (l'',k''), \\
\mathbf{0}, & \mbox{if } (l,k) \neq (l'',k'').
\end{cases}
\end{align}
The LMMSE estimate $\hat{\mathbf{h}}_{l'k'}^l$ is obtained by, first, computing the cross-covariance matrix between the two random variables $\mathbf{h}_{l'k'}^l$ and $\mathbf{y}_{l'k'}^{l,p}$ as
\begin{equation} \label{eq:CrossCov}
\mathbb{E} \Big\{ \mathbf{h}_{l'k'}^l \big(\mathbf{y}_{l'k'}^{l,p} \big)^{\rm H} \Big\} = \sqrt{\hat{p}_{l'k'}} \tau_p \beta_{l'k'}^l d_{l'k'}^l \mathbf{R}_{l'k'}^l.
\end{equation}
In fact, \eqref{eq:CrossCov} is obtained by utilizing the formulation of $\mathbf{y}_{l'k'}^{l,p}$ in \eqref{eq:ylklp} and the channel correlation property in \eqref{eq:CovMatrix}. The covariance matrix of the signal $\mathbf{y}_{l'k'}^{l,p}$ is computed as
\begin{equation} \label{eq:CovMa}
\mathbb{E} \Big\{ \mathbf{y}_{l'k'}^{l,p} \big(\mathbf{y}_{l'k'}^{l,p} \big)^{\rm H} \Big\}  = \big(\pmb{\Psi}_{l'k'}^{l}\big)^{-1} \tau_p^{-1}.
\end{equation}
By utilizing \eqref{eq:CrossCov} and \eqref{eq:CovMa} into the Bayesian Gauss-Markov theorem \cite[Theorem 12.1]{Kay1993a}, i.e.,
\begin{equation}
\hat{\mathbf{h}}_{l'k'}^{l} =  \mathbb{E} \Big\{ \mathbf{h}_{l'k'}^l \big(\mathbf{y}_{l'k'}^{l,p} \big)^{\rm H} \Big\} \left( \mathbb{E} \Big\{ \mathbf{y}_{l'k'}^{l,p} \big(\mathbf{y}_{l'k'}^{l,p} \big)^{\rm H} \Big\} \right)^{-1} \mathbf{y}_{l'k'}^{l,p}, 
\end{equation}
and doing some algebra, we obtain the expression of the channel estimate $\hat{\mathbf{h}}_{l'k'}^{l}$ as shown in the lemma. 

\subsection{Proof of Theorem~\ref{Theorem:ClosedFormMR}} \label{Appendix:ClosedFormMR}
We compute the expectation in the numerator of \eqref{eq:GeneralSINR} with noting that $\mathbf{v}_{lk}= \hat{\mathbf{h}}_{lk}^l$ as
\begin{equation} \label{eq:Numerv1}
\mathbb{E} \left\{ \mathbf{v}_{lk}^{\rm H} \mathbf{h}_{lk}^l \right\}  = \mathbb{E} \left\{ \| \mathbf{v}_{lk} \|^2 \right\} =  \sqrt{z_{lk}^l} \mathrm{tr} \left(\mathbf{R}_{lk}^l \pmb{\Psi}_{lk}^l \mathbf{R}_{lk}^l \right),
\end{equation}
where the last equality in \eqref{eq:Numerv1} is obtained by using the covariance property in \eqref{eq:CovarianceEst}. The first part of the denominator of \eqref{eq:GeneralSINR} is decomposed into the coherent and non-coherent interference based on the pilot reuse pattern as
\begin{equation} \label{eq:FirstDen}
\begin{split}
&\sum_{l'=1}^L \sum_{k'=1}^{K} p_{l'k'} \mathbb{E} \Big\{ \big| \mathbf{v}_{lk}^{\rm H} \mathbf{h}_{l'k'}^l \big|^2 \Big\} = \sum_{(l',k') \notin \mathcal{P}_{lk}}  p_{l'k'} \mathbb{E} \Big\{ \big| \mathbf{v}_{lk}^{\rm H} \mathbf{h}_{l'k'}^l \big|^2 \Big\} \\
&+  \sum_{(l',k') \in \mathcal{P}_{lk}}  p_{l'k'} \mathbb{E} \Big\{ \big| \mathbf{v}_{lk}^{\rm H} \mathbf{h}_{l'k'}^l \big|^2 \Big\}.
\end{split}
\end{equation}
The first expectation in the right-hand side of \eqref{eq:FirstDen} is non-coherent interference and computed in closed form by the independence of two random variables $\mathbf{v}_{lk}$ and $\mathbf{h}_{l'k'}^l$ as
\begin{equation} \label{eq:FirstofFirst}
\begin{split}
 &\sum_{(l',k') \notin \mathcal{P}_{lk}}  p_{l'k'} \mathbb{E} \Big\{ \big| \mathbf{v}_{lk}^{\rm H} \mathbf{h}_{l'k'}^l \big|^2 \Big\} \\
 &= \sum_{(l',k') \notin \mathcal{P}_{lk}}  p_{l'k'} \mathrm{tr} \left(  \mathbb{E} \Big\{  \mathbf{h}_{l'k'}^l  \big(\mathbf{h}_{l'k'}^l \big)^{\rm H}  \Big\} \mathbb{E} \big\{  \mathbf{v}_{lk} \mathbf{v}_{lk}^{\rm H}  \big\}  \right)  \\
&= \sum_{(l',k') \notin \mathcal{P}_{lk}}  p_{l'k'} m_{l'k'}^l  \mathrm{tr} \big( \mathbf{R}_{lk}^l \pmb{\Psi}_{lk}^l \mathbf{R}_{lk}^l \mathbf{R}_{l'k'}^l  \big).
\end{split}
\end{equation}
The second expectation in the right-hand side of \eqref{eq:FirstDen} is  coherent interference and computed by utilizing the channel estimate in \eqref{eq:ChannelEst} to construct the combining vector as
\begin{equation} \label{eq:SecondExp}
\begin{split}
&\sum_{(l',k') \in \mathcal{P}_{lk}}  p_{l'k'} \mathbb{E} \Big\{ \big| \mathbf{v}_{lk}^{\rm H} \mathbf{h}_{l'k'}^l \big|^2 \Big\} \\
& =  \sum_{(l',k') \in \mathcal{P}_{lk}}  p_{l'k'} \mathbb{E} \left\{ \Big| \big(\mathbf{y}_{lk}^{l,p} \big)^{\rm H} \big(\mathbf{B}_{lk}^l\big)^{\rm H}  \mathbf{h}_{l'k'}^l \Big|^2 \right\} \\
&= \sum_{(l',k') \in \mathcal{P}_{lk}}  p_{l'k'} \mathbb{E} \Big\{ \Big|  \sum_{(l'',k'') \in \mathcal{P}_{lk}} \sqrt{\hat{p}_{l''k''}} \tau_p \big(\mathbf{h}_{l''k''}^l\big)^{\rm H} \big(\mathbf{B}_{lk}^l\big)^{\rm H} \mathbf{h}_{l'k'}^l  \\
& \qquad \qquad  \qquad  \qquad + \pmb{\phi}_{lk}^{\rm H} \big(\mathbf{N}_l^{p}\big)^{\rm H}   \big(\mathbf{B}_{lk}^l\big)^{\rm H} \mathbf{h}_{l'k'}^l \Big|^2 \Big\} \\
& =  \sum_{(l',k') \in \mathcal{P}_{lk}} \sum_{\substack{(l'',k'') \in \mathcal{P}_{lk} \setminus (l',k') }} p_{l'k'} \hat{p}_{l''k''} \tau_p^2  \\
& \quad \times \mathbb{E} \left\{ \left| \big(\mathbf{h}_{l''k''}^l\big)^{\rm H} \big(\mathbf{B}_{lk}^l\big)^{\rm H} \mathbf{h}_{l'k'}^l  \right|^2 \right\} \\
& \quad + \sum_{(l',k') \in \mathcal{P}_{lk}} p_{l'k'} \hat{p}_{l'k'} \tau_p^2 \mathbb{E} \left\{ \Big| \big(\mathbf{h}_{l'k'}^l\big)^{\rm H} \big( \mathbf{B}_{lk}^l \big)^{\rm H} \mathbf{h}_{l'k'}^l  \Big|^2 \right\}  \\
& \quad +  \sum_{(l',k') \in \mathcal{P}_{lk}} p_{l'k'} \mathbb{E} \left\{ \Big| \pmb{\phi}_{lk}^{\rm H} \big(\mathbf{N}_l^{p}\big)^{\rm H}  \big( \mathbf{B}_{lk}^l \big)^{\rm H} \mathbf{h}_{l'k'}^l  \Big|^2 \right\},
\end{split}
\end{equation}
where $\mathbf{B}_{lk}^l = \sqrt{\hat{p}_{lk}} \beta_{lk}^l d_{lk}^l \mathbf{R}_{lk}^l \pmb{\Psi}_{lk}^l$ and the last equality in \eqref{eq:SecondExp} is decomposed based on the correlation among the channels, and the uncorrelation between the channels and noise. In the last equation of \eqref{eq:SecondExp}, the first expectation is computed by using the independence of two random variables $\mathbf{h}_{l''k''}^l$ and $\mathbf{h}_{l'k'}^l$ as
\begin{equation} \label{eq:Firstv1}
\begin{split}
&\mathbb{E} \left\{ \Big| \big(\mathbf{h}_{l''k''}^l \big)^{\rm H} \big(\mathbf{B}_{lk}^l\big)^{\rm H} \mathbf{h}_{l'k'}^l  \Big|^2 \right\} \\
&=  \beta_{l''k''}^l d_{l''k''}^l \beta_{l'k'}^l d_{l'k'}^l \mathrm{tr} \Big( \big(\mathbf{B}_{lk}^l \big)^{\rm H} \mathbf{R}_{l'k'}^l \mathbf{B}_{lk}^l \mathbf{R}_{l''k''}^l \Big) \\
&= \beta_{l''k''}^l d_{l''k''}^l \frac{m_{l'k'}^l}{\tau_p} \mathrm{tr} \big( \pmb{\Psi}_{lk}^l \mathbf{R}_{lk}^l \mathbf{R}_{l'k'}^l \mathbf{R}_{lk}^l \pmb{\Psi}_{lk}^l \mathbf{R}_{l''k''}^l \big).
\end{split}
\end{equation} 
In order to obtain the result in \eqref{eq:Firstv1}, we have borrowed \eqref{eq:Cor1eq9} in Corollary~\ref{corollary:ChannelProperties}. The second expectation of \eqref{eq:SecondExp} is computed by exploiting \eqref{eq:4moments} as
\begin{equation}
\begin{split}
& \mathbb{E} \left\{ \Big| \big(\mathbf{h}_{l'k'}^l\big)^{\rm H} \big(\mathbf{B}_{lk}^l\big)^{\rm H} \mathbf{h}_{l'k'}^l  \Big|^2 \right\} = \big(\beta_{l'k'}^l \big)^2 \left( \big(d_{l'k'}^l\big)^2 + \frac{ \mathrm{tr} \Big(\big(\widetilde{\mathbf{R}}_{l'k'}^l\big)^2 \Big)}{\big(S_{l'k'}^l \big)^2} \right)  \\
&\times \left( \left| \mathrm{tr} \Big( \mathbf{R}_{l'k'}^l  \big(\mathbf{B}_{lk}^l \big)^{\rm H}  \right)\right|^2   + \mathrm{tr} \left( \mathbf{R}_{l'k'}^l \left( \mathbf{B}_{lk}^l \right)^{\rm H} \mathbf{R}_{l'k'}^l  \mathbf{B}_{lk}^l  \right) \Big) \\
&  = \big(\beta_{l'k'}^l \big)^2 \hat{p}_{lk}   \big(\beta_{lk}^l \big)^2 \big(d_{lk}^l \big)^2\left( \big(d_{l'k'}^l\big)^2 + \frac{ \mathrm{tr} \left( \big(\widetilde{\mathbf{R}}_{l'k'}^l \big)^2 \right)}{\big(S_{l'k'}^l\big)^2} \right) \times\\
&\left( \Big| \mathrm{tr} \big(  \mathbf{R}_{l'k'}^l \pmb{\Psi}_{lk}^l  \mathbf{R}_{lk}^l  \big)\Big|^2 + \mathrm{tr} \big(  \mathbf{R}_{l'k'}^l \pmb{\Psi}_{lk}^l  \mathbf{R}_{lk}^l  \mathbf{R}_{l'k'}^l \mathbf{R}_{lk}^l \pmb{\Psi}_{lk}^l   \big) \right).
\end{split}
\end{equation}
Thanks to the independence between the channel and noise, the last expectation of \eqref{eq:SecondExp} is computed as
\begin{equation} \label{eq:Thirdv1}
\begin{split}
&\mathbb{E} \left\{ \Big| \pmb{\phi}_{lk}^{\rm H} \big(\mathbf{N}_l^{p} \big)^{\rm H}  \big(\mathbf{B}_{lk}^l \big)^{\rm H} \mathbf{h}_{l'k'}^l  \Big|^2 \right\} \\
&=  \mathrm{tr} \left( \big(\mathbf{B}_{lk}^l \big)^{\rm H} \mathbb{E} \Big\{ \mathbf{h}_{l'k'}^l \big(\mathbf{h}_{l'k'}^l \big)^{\rm H}  \Big\} \mathbf{B}_{lk}^l \mathbb{E} \Big\{ \mathbf{N}_l^{p} \pmb{\phi}_{lk}  \pmb{\phi}_{lk}^{\rm H} \big(\mathbf{N}_l^{p} \big)^{\rm H}   \Big\}    \right) \\
&=  \sigma^2 m_{l'k'}^l  \mathrm{tr} \big(  \pmb{\Psi}_{lk}^l \mathbf{R}_{lk}^l  \mathbf{R}_{l'k'}^l \mathbf{R}_{lk}^l  \pmb{\Psi}_{lk}^l \big).
\end{split}
\end{equation}
Plugging \eqref{eq:Firstv1}-\eqref{eq:Thirdv1} into \eqref{eq:SecondExp} and doing some algebra, the coherent interference term \eqref{eq:SecondExp} is obtained in closed form as
\begin{multline} \label{eq:SecondofFirst}
\begin{split}
&\sum_{(l',k') \in \mathcal{P}_{lk}}  p_{l'k'} \mathbb{E} \Big\{ \big| \mathbf{v}_{lk}^H \mathbf{h}_{l'k'}^l \big|^2 \Big\}  = \\ & \sum_{(l',k') \in \mathcal{P}_{lk}}  p_{l'k'}  m_{l'k'}^l \mathrm{tr} \big( \mathbf{R}_{lk}^l \pmb{\Psi}_{lk}^l \mathbf{R}_{lk}^l \mathbf{R}_{l'k'}^l \big) +  \sum_{(l',k') \in \mathcal{P}_{lk}} p_{l'k'} \\
& \times \frac{z_{l'k'}^l}{\big(d_{l'k'}^l\big)^2}  \left( \big(d_{l'k'}^l\big)^2 + \frac{ \mathrm{tr} \Big( \big(\widetilde{\mathbf{R}}_{l'k'}^l\big)^2 \Big)}{ \big(S_{l'k'}^l \big)^2} \right) \Big| \mathrm{tr} \big(  \mathbf{R}_{l'k'}^l \pmb{\Psi}_{lk}^l  \mathbf{R}_{lk}^l  \big)\Big|^2   + \\
& \sum_{(l',k') \in \mathcal{P}_{lk}} p_{l'k'} z_{l'k'}^l\frac{ \mathrm{tr} \Big( \big(\widetilde{\mathbf{R}}_{l'k'}^l \big)^2 \Big)}{ \big(S_{l'k'}^l\big)^2}  
 \mathrm{tr} \left(  \mathbf{R}_{l'k'}^l \pmb{\Psi}_{lk}^l  \mathbf{R}_{lk}^l  \mathbf{R}_{l'k'}^l \mathbf{R}_{lk}^l \pmb{\Psi}_{lk}^l   \right).
\end{split}
\end{multline}
Combining \eqref{eq:FirstDen}, \eqref{eq:FirstofFirst}, and \eqref{eq:SecondofFirst}, the first part of the denominator of \eqref{eq:GeneralSINR} is computed in closed form as
\begin{equation} \label{eq:FirstFirstv1}
\begin{split}
& \sum_{l'=1}^L \sum_{k'=1}^{K} p_{l'k'} \mathbb{E} \Big\{ \big| \mathbf{v}_{lk}^{\rm H} \mathbf{h}_{l'k'}^l \big|^2 \Big\} = \sum_{l'=1}^L \sum_{k'=1}^{K}  p_{l'k'}  m_{l'k'}^l \times \\
&  \mathrm{tr} \big( \mathbf{R}_{lk}^l \pmb{\Psi}_{lk}^l \mathbf{R}_{lk}^l \mathbf{R}_{l'k'}^l \big) + \sum_{(l',k') \in \mathcal{P}_{lk}} p_{l'k'} \frac{z_{l'k'}^l}{\big(d_{l'k'}^l\big)^2} \times \\ 
&    
 \left( \big(d_{l'k'}^l\big)^2 + \frac{ \mathrm{tr} \Big( \big(\widetilde{\mathbf{R}}_{l'k'}^l\big)^2 \Big)}{ \big(S_{l'k'}^l \big)^2} \right) \Big| \mathrm{tr} \big(  \mathbf{R}_{l'k'}^l \pmb{\Psi}_{lk}^l  \mathbf{R}_{lk}^l  \big)\Big|^2 +  \sum_{(l',k') \in \mathcal{P}_{lk}} p_{l'k'}  \\
 & \times z_{l'k'}^l\frac{ \mathrm{tr} \Big( \big(\widetilde{\mathbf{R}}_{l'k'}^l \big)^2 \Big)}{ \big(S_{l'k'}^l\big)^2} \mathrm{tr} \left(  \mathbf{R}_{l'k'}^l \pmb{\Psi}_{lk}^l  \mathbf{R}_{lk}^l  \mathbf{R}_{l'k'}^l \mathbf{R}_{lk}^l \pmb{\Psi}_{lk}^l   \right).
\end{split}
\end{equation}
Utilizing \eqref{eq:Numerv1} and \eqref{eq:FirstFirstv1} into \eqref{eq:GeneralSINR} together with doing some algebra, we obtain the closed-form SINR expression as in the theorem. 
\subsection{Proof of Theorem~\ref{theorem:Asymptotic}} \label{Appendix:Asymptotic}
We begin with dividing the numerator and denominator of the SINR expression \eqref{eq:SINRlkMR} by $M  \mathrm{tr} \big(\mathbf{R}_{lk}^l \pmb{\Psi}_{lk}^l \mathbf{R}_{lk}^l \big)$. The numerator of \eqref{eq:SINRlkMR} is $  p_{lk} z_{lk}^l  \mathrm{tr} \left(\mathbf{R}_{lk}^l \pmb{\Psi}_{lk}^l \mathbf{R}_{lk}^l \right)/M$. 
Meanwhile, the first part in the denominator of \eqref{eq:SINRlkMR} becomes
\begin{equation} \label{eq:FirstTerm}
\begin{split}
\frac{\mathsf{NI}_{lk}}{M  \mathrm{tr} \big(\mathbf{R}_{lk}^l \pmb{\Psi}_{lk}^l \mathbf{R}_{lk}^l \big)} &= \frac{\sum_{l' =1 }^L \sum_{k'=1}^{K}  p_{l'k'} m_{l'k'}^l   \mathrm{tr} \big( \mathbf{R}_{lk}^l \pmb{\Psi}_{lk}^l \mathbf{R}_{lk}^l \mathbf{R}_{l'k'}^l  \big)}{M  \mathrm{tr} \big(\mathbf{R}_{lk}^l \pmb{\Psi}_{lk}^l \mathbf{R}_{lk}^l \big)} \\
&\stackrel{(a)}{\leq} \frac{1}{M}  \sum_{l' =1 }^L \sum_{k'=1}^{K}  p_{l'k'} m_{l'k'}^l \big\| \mathbf{R}_{l'k'}^l \big\|_2  \\
& \leq \frac{LK}{M}  \underset{(l',k')}{\max} \, \,p_{l'k'} m_{l'k'}^l \big\| \mathbf{R}_{l'k'}^l \big\|_2,
\end{split}
\end{equation}
where $(a)$ is obtained by the upper bound of the trace matrix expression \cite[Lemma B.7]{Bjornson2017bo}. By applying Assumption~\ref{Assumption1} to the last result \eqref{eq:FirstTerm}, we observe that this part converges to zero as either $M \rightarrow \infty$ or $S_{l'k'}^l \rightarrow \infty$. It is also straightforward to prove that the last part in the denominator of the SINR expression \eqref{eq:SINRlkMR} converges to zero as either $M \rightarrow \infty$ or $S_{l'k'}^l \rightarrow \infty$, i.e.,
\begin{equation} \label{eq:LastTerm}
\frac{\mathsf{NI}_{lk}}{M  \mathrm{tr} \big(\mathbf{R}_{lk}^l \pmb{\Psi}_{lk}^l \mathbf{R}_{lk}^l \big)} \rightarrow 0.
\end{equation}
Combining \eqref{eq:FirstTerm} and \eqref{eq:LastTerm}, the denominator of \eqref{eq:SINRlkMR} is formulated as $\mathsf{CI}_{lk}$, and therefore the asymptotic SINR expression as $M \rightarrow \infty$ for a given finite set of the scatterers and covariance matrices as shown in \eqref{eq:AsymptRatev1}. 

When $\mathbf{R}_{lk}^l$ is asymptotically orthogonal with all the other covariance matrices of the users sharing the same pilot signal as user~$k$ in cell~$l$, the second part in the denominator of \eqref{eq:SINRlkMR} converges to as
\begin{equation}
\begin{split}
 & \frac{\mathsf{CI}_{lk}}{M  \mathrm{tr} \big(\mathbf{R}_{lk}^l \pmb{\Psi}_{lk}^l \mathbf{R}_{lk}^l \big)} \rightarrow   \frac{p_{lk} z_{lk}^l  \mathrm{tr} \Big(\big(\widetilde{\mathbf{R}}_{lk}^l\big)^2 \Big)  \mathrm{tr} \big(\mathbf{R}_{lk}^l \pmb{\Psi}_{lk}^l \mathbf{R}_{lk}^l \big) }{ M \big( d_{lk}^l S_{lk}^l \big)^2 }  \\
& + \frac{p_{lk} z_{lk}^l  \mathrm{tr} \Big(\big(\widetilde{\mathbf{R}}_{lk}^l\big)^2 \Big)  \mathrm{tr} \Big(\big(\mathbf{R}_{lk}^l \pmb{\Psi}_{lk}^l \mathbf{R}_{lk}^l \big)^2 \Big) }{ M \big( d_{lk}^l S_{lk}^l \big)^2 \mathrm{tr} \big(\mathbf{R}_{lk}^l \pmb{\Psi}_{lk}^l \mathbf{R}_{lk}^l \big)}\\
 \stackrel{(a)}{\leq} & \frac{p_{lk} z_{lk}^l  \mathrm{tr} \Big(\big(\widetilde{\mathbf{R}}_{lk}^l\big)^2 \Big)  \mathrm{tr} \big(\mathbf{R}_{lk}^l \pmb{\Psi}_{lk}^l \mathbf{R}_{lk}^l \big) }{ M \big( d_{lk}^l S_{lk}^l \big)^2 } \\
& + \frac{p_{lk} z_{lk}^l  \mathrm{tr} \Big(\big(\widetilde{\mathbf{R}}_{lk}^l\big)^2 \Big)  \big\| \mathbf{R}_{lk}^l \pmb{\Psi}_{lk}^l \mathbf{R}_{lk}^l  \Big\|_2 }{ M \big( d_{lk}^l S_{lk}^l \big)^2} \\
& \xrightarrow{(b)} \frac{p_{lk} z_{lk}^l  \mathrm{tr} \Big(\big(\widetilde{\mathbf{R}}_{lk}^l\big)^2 \Big)  \mathrm{tr} \big(\mathbf{R}_{lk}^l \pmb{\Psi}_{lk}^l \mathbf{R}_{lk}^l \big) }{ M \big( d_{lk}^l S_{lk}^l \big)^2 },
\end{split}
\end{equation}
where $(a)$ is obtained by \cite[Lemma B.7]{Bjornson2017bo} and $(b)$ is because of our assumptions on the covariance matrices. Consequently, the asymptotic uplink SE of user~$k$ in cell~$l$ is obtained as in \eqref{eq:Limitb}.

As both the number of antennas at each BS and scatterers go without bound while the covariance matrices are non-orthogonal, the first and last parts in the denominator of \eqref{eq:SINRlkMR} go to zeros, while the second part converges to as
\begin{multline}
\frac{\mathsf{CI}_{lk}}{M  \mathrm{tr} \big(\mathbf{R}_{lk}^l \pmb{\Psi}_{lk}^l \mathbf{R}_{lk}^l \big)} \rightarrow \\ \frac{\sum_{(l',k') \in \mathcal{P}_{lk} \setminus (l,k)} p_{l'k'} z_{l'k'}^l \Big| \mathrm{tr} \big(  \mathbf{R}_{l'k'}^l \pmb{\Psi}_{lk}^l  \mathbf{R}_{lk}^l  \big)\Big|^2}{M  \mathrm{tr} \big(\mathbf{R}_{lk}^l \pmb{\Psi}_{lk}^l \mathbf{R}_{lk}^l \big)},
\end{multline}
and hence we obtain the asymptotic SE expression as shown in \eqref{eq:AsymptRatev3}. For the last case in \eqref{eq:AsymptRatev4} is obtained since the denominator of \eqref{eq:SINRlkMR} goes to zeros, while the numerator goes to a constant.
\subsection{Proof of Theorem~\ref{Theorem:Alg1}} \label{Appendix:Alg1}
We first prove that every $I_{lk}(\mathbf{p})$ is a standard interference function as given in Definition~\ref{Def:TypeI}. Indeed, the positivity property is true since it holds for all $\mathbf{p} \succeq \mathbf{0}$ that
\begin{equation}
\begin{split}
& I_{lk}(\mathbf{p}) \geq I_{lk}(\mathbf{0}) \stackrel{(a)}{=} \frac{\nu_{lk} \mathsf{NO}_{lk}}{z_{lk}^l \left| \mathrm{tr} \left(\mathbf{R}_{lk}^l \pmb{\Psi}_{lk}^l \mathbf{R}_{lk}^l \right) \right|^2}  \\
& \stackrel{(b)}{=} \frac{\sigma^2}{\hat{p}_{lk} (\beta_{lk}^l)^2 (d_{lk}^l)^2 \tau_p  \mathrm{tr} \left(\mathbf{R}_{lk}^l \pmb{\Psi}_{lk}^l \mathbf{R}_{lk}^l \right)} > 0,
\end{split}
\end{equation}
where $(a)$ is obtained since $\mathsf{NO}_{lk}$ is independent of the data powers and $(b)$ is obtained after doing some algebra. Let us denote the two vectors $\mathbf{p}$ and $\mathbf{p}'$ having $p_{lk} \geq p_{lk}', \forall l,k$, then we obtain
\begin{multline}
I_{lk}(\mathbf{p}) - I_{lk}(\mathbf{p}') = \\ \frac{\nu_{lk} \left(\mathsf{NI}_{lk}\big(\mathbf{p} \big) - \mathsf{NI}_{lk}\big(\mathbf{p}'\big)  \right) + \nu_{lk} \left(\mathsf{CI}_{lk}\big(\mathbf{p} \big) - \mathsf{CI}_{lk} \big(\mathbf{p}' \big) \right)  }{z_{lk}^l \left| \mathrm{tr} \left(\mathbf{R}_{lk}^l \pmb{\Psi}_{lk}^l \mathbf{R}_{lk}^l \right) \right|^2} \geq 0,
\end{multline}
which means $I_{lk}(\mathbf{p}) \geq I_{lk}(\mathbf{p}')$ and confirms the monotonicity. For the scalability, we observe that
\begin{equation}
\begin{split}
\alpha I_{lk}(\mathbf{p}) &=  \frac{\alpha \nu_{lk} \mathsf{NI}_{lk} (\mathbf{p}) + \alpha  \nu_{lk} \mathsf{CI}_{lk} (\mathbf{p}) + \alpha \nu_{lk} \mathsf{NO}_{lk} }{ z_{lk}^l \left| \mathrm{tr} \left(\mathbf{R}_{lk}^l \pmb{\Psi}_{lk}^l \mathbf{R}_{lk}^l \right) \right|^2 } \\
 &\stackrel{(a)}{=} \frac{ \nu_{lk} \mathsf{NI}_{lk} ( \alpha \mathbf{p}) +   \nu_{lk} \mathsf{CI}_{lk} (\alpha \mathbf{p}) + \alpha \nu_{lk} \mathsf{NO}_{lk} }{ z_{lk}^l \left| \mathrm{tr} \left(\mathbf{R}_{lk}^l \pmb{\Psi}_{lk}^l \mathbf{R}_{lk}^l \right) \right|^2 }\\
& \stackrel{(b)}{\geq} \frac{ \nu_{lk} \mathsf{NI}_{lk} ( \alpha \mathbf{p}) +   \nu_{lk} \mathsf{CI}_{lk} (\alpha \mathbf{p}) +  \nu_{lk} \mathsf{NO}_{lk} }{ z_{lk}^l \left| \mathrm{tr} \left(\mathbf{R}_{lk}^l \pmb{\Psi}_{lk}^l \mathbf{R}_{lk}^l \right) \right|^2 }\\
& = I_{lk}(\alpha \mathbf{p}),
 \end{split}
\end{equation}
which confirms that $I_{lk}(\mathbf{p})$ satisfies the monotonicity property. Since every $I_{lk}(\mathbf{p})$ is a  interference function, the update procedure in \eqref{eq:plkn} guarantees: First, beginning with the initial data power values $p_{lk}(0) = P_{\max,lk}, \forall l,k,$ all the updated power coefficients at iteration~$n$ are in the feasible domain. Indeed, we can prove this statement by mathematical induction following similar steps as \cite[Lemma~3]{van2020uplink}. Second, the update in \eqref{eq:plkn} ensures a reduction of the objective function along iterations.
\subsection{Proof of Theorem~\ref{theorem:Standardfunction}} \label{Appendix:Standardfunction}
Before getting in the proof, we recall the so-called two-sided
function \cite{sung2005generalized}. Specifically, a function $f(\mathbf{z})$ is a two-sided scalable if for $\forall \alpha > 1$ and $\frac{1}{\alpha} \mathbf{z} \preceq \hat{\mathbf{z}} \preceq \alpha \mathbf{z},$ implies the following two-sided inequality
\begin{equation}
	\frac{1}{\alpha} f (\mathbf{z}) < f(\hat{\mathbf{z}}) <  \alpha f(\mathbf{z}).
\end{equation}
We stress that the authors in \cite{rasti2010distributed} gave a toy example of a two-sided scalable function to update the data transmit power for a communication system under perfect channel state information. Unlike the previous works, all the functions $f_{lk} \left( \mathbf{p}(n-1) \right)$ involve the complicated expressions of many effects from channel estimation, pilot contamination, non-coherent interference, and noise.

We now prove that $f_{lk} \left( \mathbf{p}(n-1) \right)$ is a two-sided scalable function. If $I_{lk} \left( \mathbf{p}(n-1) \right) \leq P_{\max,lk}$, then it is sufficient to prove that $I_{lk} \left( \mathbf{p}(n-1) \right)$ is a two-side scalable function. Indeed, we have shown in Theorem~\ref{Algorithm1} that $I_{lk} \left( \mathbf{p}(n-1) \right)$ is a standard interference function. Therefore, for $\frac{1}{\alpha} \mathbf{p}_l(n-1) \preceq \hat{\mathbf{p}}(n-1) \preceq \alpha \mathbf{p}_l (n-1)$, we have:
\begin{equation} \label{eq:ProofTwosidedeq1}
I_{lk} \left( \mathbf{p}(n-1) \right) \stackrel{(a)}{<} I_{lk} \left( \alpha \hat{\mathbf{p}}(n-1) \right) \stackrel{(b)}{<} \alpha I_{lk} \left(\hat{\mathbf{p}}(n-1) \right),
\end{equation}
where $(a)$ is obtained by applying the monotonicity property for $\mathbf{p}(n-1) \preceq \alpha \hat{\mathbf{p}}(n-1)$; $(b)$ is obtained by using the scalability property for $\alpha \hat{\mathbf{p}}(n-1)$. As a consequence of \eqref{eq:ProofTwosidedeq1}, 
\begin{equation} \label{eq:ProofSide1}
\frac{1}{\alpha} I_{lk} \left( \mathbf{p}(n-1) \right) < I_{lk} \left(\hat{\mathbf{p}}(n-1) \right).
\end{equation}
Similarly, by applying the monotonicity and scalability properties for $\hat{\mathbf{p}}(n-1) \preceq \alpha \mathbf{p}(n-1)$, the following inequalities are obtained as
\begin{equation}
I_{lk} \left( \hat{\mathbf{p}}(n-1) \right) < I_{lk} \left( \alpha \mathbf{p}(n-1) \right) < \alpha I_{lk} \left( \mathbf{p}(n-1) \right) ,
\end{equation}
which results in
\begin{equation} \label{eq:ProofSide2}
I_{lk} \left( \hat{\mathbf{p}}(n-1) \right) < \alpha I_{lk} \left( \mathbf{p}(n-1) \right).
\end{equation}
Combining \eqref{eq:ProofSide1} and \eqref{eq:ProofSide2}, we attain the two-sided scalable property of $I_{lk} \left( \hat{\mathbf{p}}(n-1) \right)$ as
\begin{equation} \label{eq:ProofTwoSided}
\frac{1}{\alpha} I_{lk} \left( \mathbf{p}(n-1) \right) < I_{lk} \left( \hat{\mathbf{p}}(n-1) \right) < \alpha I_{lk} \left( \mathbf{p}(n-1) \right).
\end{equation}
We now prove that $P_{\max,lk}^2 / I_{lk} \left( \hat{\mathbf{p}}(n-1) \right)$ is also a two-side scalable function. In fact, this is straightforward since $I_{lk} \left( \hat{\mathbf{p}}(n-1) \right)$ satisfies the positivity, an inversion of \eqref{eq:ProofTwoSided} is
\begin{equation} \label{eq:ProofTwoSidedv2}
\frac{1}{\alpha} \frac{1}{I_{lk} \left(\mathbf{p}(n-1) \right)} <\frac{1}{I_{lk} \left( \hat{\mathbf{p}}(n-1) \right)} < \alpha \frac{1}{I_{lk} \left( \mathbf{p}(n-1) \right)}.
\end{equation}
Multiplying \eqref{eq:ProofTwoSidedv2} by $P_{\max,lk}^2$, we obtain the following inequalities
\begin{equation} \label{eq:ProofTwoSidedv3}
\frac{1}{\alpha} \frac{P_{\max,lk}^2}{I_{lk} \left(\mathbf{p}(n-1) \right)} <\frac{P_{\max,lk}^2}{I_{lk} \left( \hat{\mathbf{p}}(n-1) \right)} < \alpha \frac{P_{\max,lk}^2}{I_{lk} \left( \mathbf{p}(n-1) \right)},
\end{equation}
which completes the proof that confirms $f_{lk}(\mathbf{p}(n-1))$ being a two-side scalable function. From the initial values $p_{lk}(0) = P_{\max,lk}, \forall l,k,$ the update in \eqref{eq:UpdatedPowerv1} ensures that the iterative algorithm will converge to a fixed point.  
\bibliographystyle{IEEEtran}
\bibliography{IEEEabrv,refs}

\begin{thebibliography}{10}
\providecommand{\url}[1]{#1}
\csname url@samestyle\endcsname
\providecommand{\newblock}{\relax}
\providecommand{\bibinfo}[2]{#2}
\providecommand{\BIBentrySTDinterwordspacing}{\spaceskip=0pt\relax}
\providecommand{\BIBentryALTinterwordstretchfactor}{4}
\providecommand{\BIBentryALTinterwordspacing}{\spaceskip=\fontdimen2\font plus
\BIBentryALTinterwordstretchfactor\fontdimen3\font minus
  \fontdimen4\font\relax}
\providecommand{\BIBforeignlanguage}[2]{{%
\expandafter\ifx\csname l@#1\endcsname\relax
\typeout{** WARNING: IEEEtran.bst: No hyphenation pattern has been}%
\typeout{** loaded for the language `#1'. Using the pattern for}%
\typeout{** the default language instead.}%
\else
\language=\csname l@#1\endcsname
\fi
#2}}
\providecommand{\BIBdecl}{\relax}
\BIBdecl

\bibitem{Chien2021ICC}
T.~V. Chien, H.~Q. Ngo, S.~Chatzinotas, B.~Ottersten, and M.~Debbah, ``Massive
  {MIMO} under double scattering channels: {P}ower minimization and congestion
  controls,'' in \emph{Proc.~IEEE ICC}, 2021.

\bibitem{tariq2020speculative}
F.~Tariq, M.~R. Khandaker, K.-K. Wong, M.~A. Imran, M.~Bennis, and M.~Debbah,
  ``A speculative study on {6G},'' \emph{{IEEE} Wireless Commun. Mag.},
  vol.~27, no.~4, pp. 118--125, 2020.

\bibitem{nguyen2020performance}
T.~H. Nguyen, W.-S. Jung, L.~T. Tu, T.~Van~Chien, D.~Yoo, and S.~Ro,
  ``Performance analysis and optimization of the coverage probability in dual
  hop {LoRa} networks with different fading channels,'' \emph{IEEE Access},
  vol.~8, pp. 107\,087--107\,102, 2020.

\bibitem{index2019global}
C.~V.~N. Index, ``Cisco visual networking index: {G}lobal mobile data traffic
  forecast update, 2017-2022 white paper,'' \emph{Cisco: San Jose, CA, USA},
  2019.

\bibitem{bjornson2019massive}
E.~Bj{\"o}rnson, L.~Sanguinetti, H.~Wymeersch, J.~Hoydis, and T.~L. Marzetta,
  ``Massive {MIMO} is a reality—what is next?: {F}ive promising research
  directions for antenna arrays,'' \emph{Digital Signal Processing}, vol.~94,
  pp. 3--20, 2019.

\bibitem{Marzetta2010a}
T.~L. Marzetta, ``Noncooperative cellular wireless with unlimited numbers of
  base station antennas,'' \emph{{IEEE} Trans. Wireless Commun.}, vol.~9,
  no.~11, pp. 3590--3600, 2010.

\bibitem{Hoydis2013a}
J.~Hoydis, S.~ten Brink, and M.~Debbah, ``Massive {MIMO} in the {UL/DL} of
  cellular networks: How many antennas do we need?'' \emph{{IEEE} J. Sel. Areas
  Commun.}, vol.~31, no.~2, pp. 160--171, 2013.

\bibitem{Ngo2013a}
H.~Q. Ngo, E.~G. Larsson, and T.~L. Marzetta, ``Energy and spectral efficiency
  of very large multiuser {MIMO} systems,'' \emph{{IEEE} Trans. Commun.},
  vol.~61, no.~4, pp. 1436--1449, 2013.

\bibitem{Chien2016b}
T.~V. Chien, E.~Bj{\"o}rnson, and E.~G. Larsson, ``Joint power allocation and
  user association optimization for {M}assive {MIMO} systems,'' \emph{{IEEE}
  Trans. Wireless Commun.}, vol.~15, no.~9, pp. 6384 -- 6399, 2016.

\bibitem{Chien2018a}
T.~Van~Chien, C.~Moll{\'e}n, and E.~Bj{\"o}rnson, ``Large-scale-fading decoding
  in cellular {M}assive {MIMO} systems with spatially correlated channels,''
  \emph{{IEEE} Trans. Commun.}, vol.~67, no.~4, pp. 2746 -- 2762, 2019.

\bibitem{kammoun2014linear}
A.~Kammoun, A.~M{\"u}ller, E.~Bj{\"o}rnson, and M.~Debbah, ``Linear precoding
  based on polynomial expansion: {L}arge-scale multi-cell {MIMO} systems,''
  \emph{{IEEE} J. Sel. Areas Commun.}, vol.~8, no.~5, pp. 861--875, 2014.

\bibitem{kammoun2015generalized}
Q.-U.-A. Nadeem, A.~Kammoun, M.~Debbah, and M.-S. Alouini, ``A generalized
  spatial correlation model for {3D} {MIMO} channels based on the fourier
  coefficients of power spectrums,'' \emph{{IEEE} Trans. Signal Process.},
  vol.~63, no.~14, pp. 3671--3686, 2015.

\bibitem{yu2002models}
K.~Yu and B.~Ottersten, ``Models for {MIMO} propagation channels: {A} review,''
  \emph{Wireless communications and mobile computing}, vol.~2, no.~7, pp.
  653--666, 2002.

\bibitem{Gesbert2002a}
D.~Gesbert, H.~B{\"o}lcskei, D.~Gore, and A.~Paulraj, ``Outdoor {MIMO} wireless
  channels: Models and performance prediction,'' \emph{{IEEE} Trans. Commun.},
  vol.~50, no.~12, pp. 1926--1934, 2002.

\bibitem{ngo2017no}
H.~Q. Ngo and E.~G. Larsson, ``No downlink pilots are needed in {TDD} massive
  {MIMO},'' \emph{{IEEE} Trans. Wireless Commun.}, vol.~16, no.~5, pp.
  2921--2935, 2017.

\bibitem{James2020}
J.~A.~C. Sutton, H.~Q. Ngo, and M.~Matthaiou, ``Hardening the channels by
  precoder design in {M}assive {MIMO} with multiple-antenna users,''
  \emph{{IEEE} Trans. Veh. Technol.}, 2020.

\bibitem{van2016multi}
T.~Van~Chien, E.~Bj{\"o}rnson, and E.~G. Larsson, ``Multi-cell {M}assive {MIMO}
  performance with double scattering channels,'' in \emph{Proc.~IEEE
  CAMAD}.\hskip 1em plus 0.5em minus 0.4em\relax IEEE, 2016, pp. 231--236.

\bibitem{kammoun2019asymptotic}
A.~Kammoun, M.~Debbah, and M.-S. Alouini, ``Asymptotic analysis of {RZF} over
  double scattering channels with {MMSE} estimation,'' \emph{{IEEE} Trans.
  Wireless Commun.}, vol.~18, no.~5, pp. 2509--2526, 2019.

\bibitem{Ye2020}
J.~{Ye}, Q.~{Nadeem}, A.~{Kammoun}, and M.~{Alouini}, ``Asymptotic analysis of
  {MRT} over double scattering channels with {MMSE} estimation,'' \emph{{IEEE}
  Trans. Wireless Commun.}, vol.~19, no.~12, pp. 7851--7863, 2020.

\bibitem{massivemimobook}
\BIBentryALTinterwordspacing
E.~Bj\"{o}rnson, J.~Hoydis, and L.~Sanguinetti, ``Massive {MIMO} networks:
  {Spectral}, energy, and hardware efficiency,'' \emph{Foundations and
  Trends{\textregistered} in Signal Processing}, vol.~11, no. 3-4, pp.
  154--655, 2017. [Online]. Available:
  \url{http://dx.doi.org/10.1561/2000000093}
\BIBentrySTDinterwordspacing

\bibitem{Chien2017b}
T.~V. Chien, E.~Bj{\"o}rnson, and E.~G. Larsson, ``Joint pilot sequence design
  and power control for max-min fairness in uplink {M}assive {MIMO},'' in
  \emph{Proc.~IEEE ICC}, 2017.

\bibitem{Jin2015a}
S.~Jin, M.~Li, Y.~Huang, Y.~Du, and X.~Gao, ``Pilot scheduling schemes for
  multi-cell massive multiple-input-multiple-output transmission,'' \emph{IET
  Communications}, vol.~9, no.~5, pp. 689--700, 2015.

\bibitem{nguyen2018optimal}
T.~H. Nguyen, T.~K. Nguyen, H.~D. Han, and V.~D. Nguyen, ``Optimal power
  control and load balancing for uplink cell-free multi-user {M}assive
  {MIMO},'' \emph{{IEEE} Access}, vol.~6, pp. 14\,462--14\,473, 2018.

\bibitem{Sun2015}
R.~Sun, M.~Hong, and Z.~Q. Luo, ``Joint downlink base station association and
  power control for max-min fairness: Computation and complexity,''
  \emph{{IEEE} J. Sel. Areas Commun.}, vol.~33, no.~6, pp. 1040--1054, 2015.

\bibitem{ghazanfari2019fair}
A.~Ghazanfari, H.~V. Cheng, E.~Bj{\"o}rnson, and E.~G. Larsson, ``A fair and
  scalable power control scheme in multi-cell {M}assive {MIMO},'' in
  \emph{Proc.~ICASSP}.\hskip 1em plus 0.5em minus 0.4em\relax IEEE, 2019, pp.
  4499--4503.

\bibitem{van2020joint}
T.~Van~Chien, E.~Bj{\"o}rnson, and E.~G. Larsson, ``Optimal design of
  energy-efficient {C}ell-free {M}assive {MIMO}: {J}oint power allocation and
  load balancing,'' in \emph{Proc.~IEEE ICASSP}, 2019.

\bibitem{ngo2018total}
H.~Q. Ngo, L.-N. Tran, T.~Q. Duong, M.~Matthaiou, and E.~G. Larsson, ``On the
  total energy efficiency of cell-free {M}assive {MIMO},'' \emph{IEEE
  Transactions on Green Communications and Networking}, vol.~2, no.~1, pp.
  25--39, 2018.

\bibitem{van2020uplink}
T.~Van~Chien, E.~Bj{\"o}rnson, and H.~Q. Ngo, ``Uplink power control in
  cellular {M}assive {MIMO} systems: {C}oping with the congestion issue,'' in
  \emph{Proc. Of ICC Workshops}.\hskip 1em plus 0.5em minus 0.4em\relax IEEE,
  2020, pp. 1--6.

\bibitem{patzold2003mobile}
M.~Patzold, \emph{Mobile fading channels}.\hskip 1em plus 0.5em minus
  0.4em\relax John Wiley \& Sons, Inc., 2003.

\bibitem{Kay1993a}
S.~Kay, \emph{Fundamentals of Statistical Signal Processing: Estimation
  Theory}.\hskip 1em plus 0.5em minus 0.4em\relax Prentice Hall, 1993.

\bibitem{Marzetta2016a}
T.~L. Marzetta, E.~G. Larsson, H.~Yang, and H.~Q. Ngo, \emph{Fundamentals of
  {M}assive {MIMO}}.\hskip 1em plus 0.5em minus 0.4em\relax Cambridge
  University Press, 2016.

\bibitem{van2020massivechannels}
T.~V. Chien and H.~Q. Ngo, ``Massive {MIMO} channels,'' in \emph{Antennas and
  Propagation for 5G and Beyond}.\hskip 1em plus 0.5em minus 0.4em\relax IET
  Publisher, 2020, pp. 1--33.

\bibitem{Bjornson2017bo}
E.~Bj\"{o}rnson, J.~Hoydis, and L.~Sanguinetti, ``{M}assive {MIMO} networks:
  {S}pectral, energy, and hardware efficiency,'' \emph{Foundations and Trends
  in Signal Processing}, vol.~11, no. 3-4, pp. 154 -- 655, 2017.

\bibitem{Yin2013a}
H.~Yin, D.~Gesbert, M.~Filippou, and Y.~Liu, ``A coordinated approach to
  channel estimation in large-scale multiple-antenna systems,'' \emph{{IEEE} J.
  Sel. Areas Commun.}, vol.~31, no.~2, pp. 264--273, 2013.

\bibitem{ozdogan2019massive}
{\"O}.~{\"O}zdogan, E.~Bj{\"o}rnson, and E.~G. Larsson, ``{M}assive {MIMO} with
  spatially correlated {R}ician fading channels,'' \emph{{IEEE} Trans.
  Commun.}, vol.~67, no.~5, pp. 3234--3250, 2019.

\bibitem{Boyd2004a}
S.~Boyd and L.~Vandenberghe, \emph{Convex Optimization}.\hskip 1em plus 0.5em
  minus 0.4em\relax Cambridge University Press, 2004.

\bibitem{cvx2015}
{CVX Research Inc.}, ``{CVX}: Matlab software for disciplined convex
  programming, academic users,'' \url{http://cvxr.com/cvx}, 2015.

\bibitem{Yates1995a}
R.~Yates, ``A framework for uplink power control in cellular radio systems,''
  \emph{{IEEE} J. Sel. Areas Commun.}, vol.~13, no.~7, pp. 1341--1347, 1995.

\bibitem{LTE2010a}
\emph{Evolved Universal Terrestrial Radio Access ({E-UTRA}); {Physical}
  Channels and Modulation ({R}elease 9)}.\hskip 1em plus 0.5em minus
  0.4em\relax {3GPP} {TS} 36.213, Sep. 2010.

\bibitem{senel2019joint}
K.~Senel, E.~Bj{\"o}rnson, and E.~G. Larsson, ``Joint transmit and circuit
  power minimization in {M}assive {MIMO} with downlink {SINR} constraints:
  {W}hen to turn on {M}assive {MIMO}?'' \emph{{IEEE} Trans. Wireless Commun.},
  vol.~18, no.~3, pp. 1834 -- 1846, 2019.

\bibitem{van2020power}
T.~Van~Chien, E.~Bj{\"o}rnson, and E.~G. Larsson, ``Joint power allocation and
  load balancing optimization for energy-efficient {C}ell-free {M}assive {MIMO}
  networks,'' \emph{{IEEE} Trans. Wireless Commun.}, vol.~19, no.~10, pp. 6798
  -- 6812, 2020.

\bibitem{sung2005generalized}
C.~W. Sung and K.-K. Leung, ``A generalized framework for distributed power
  control in wireless networks,'' \emph{{IEEE} Trans. Inf. Theory}, vol.~51,
  no.~7, pp. 2625--2635, 2005.

\bibitem{rasti2010distributed}
M.~Rasti and A.~R. Sharafat, ``Distributed uplink power control with soft
  removal for wireless networks,'' \emph{{IEEE} Trans. Commun.}, vol.~59,
  no.~3, pp. 833--843, 2010.

\end{thebibliography}

\begin{IEEEbiography} 
	{Trinh Van Chien} (S'16-M'20) received the B.S. degree in Electronics and Telecommunications from Hanoi University of Science and Technology (HUST), Vietnam, in 2012. He then received the M.S. degree in Electrical and Computer Enginneering from Sungkyunkwan University (SKKU), Korea, in 2014 and the Ph.D. degree in Communication Systems from Link\"oping University (LiU), Sweden, in 2020. He is now a research associate at University of Luxembourg. His interest lies in convex optimization problems and machine learning applications for wireless communications and image \& video processing. He was an IEEE wireless communications letters exemplary reviewer for 2016 and 2017. He also received the award of scientific excellence in the first year of the 5Gwireless project funded by European Union Horizon's 2020.
\end{IEEEbiography}
\begin{IEEEbiography}
	{Hien Quoc Ngo}  received the B.S. degree in electrical engineering from the Ho Chi Minh City University of Technology, Vietnam, in 2007, the M.S. degree in electronics and radio engineering from Kyung Hee University, South Korea, in 2010, and the Ph.D. degree in communication systems from Link\"oping University (LiU), Sweden, in 2015. In 2014, he visited the Nokia Bell Labs, Murray Hill, New Jersey, USA. From January 2016 to April 2017, Hien Quoc Ngo was a VR researcher at the Department of Electrical Engineering (ISY), LiU. He was also a Visiting Research Fellow at the School of Electronics, Electrical Engineering and Computer Science, Queen's University Belfast, UK, funded by the Swedish Research Council.
	
	Hien Quoc Ngo is currently a Reader (Associate Professor) at Queen's University Belfast, UK. His main research interests include massive (large-scale) MIMO systems, cell-free massive MIMO, physical layer security, and cooperative communications. He has co-authored many research papers in wireless communications and co-authored the Cambridge University Press textbook \emph{Fundamentals of Massive MIMO} (2016).
	
	Dr. Hien Quoc Ngo received the IEEE ComSoc Stephen O. Rice Prize in Communications Theory in 2015, the IEEE ComSoc Leonard G. Abraham Prize in 2017, and the Best PhD Award from EURASIP in 2018. He also received the IEEE Sweden VT-COM-IT Joint Chapter Best Student Journal Paper Award in 2015. He was an \emph{IEEE Communications Letters} exemplary reviewer for 2014, an \emph{IEEE Transactions on Communications} exemplary reviewer for 2015, and an \emph{IEEE Wireless Communications Letters} exemplary reviewer for 2016.  He was awarded the UKRI Future Leaders Fellowship in 2019.
	Dr. Hien Quoc Ngo currently serves as an Editor for the IEEE Transactions on Wireless Communications, IEEE Wireless Communications Letters, Digital Signal Processing, Elsevier Physical Communication (PHYCOM), and IEICE Transactions on Fundamentals of Electronics, Communications and Computer Sciences. He was a Guest Editor of IET Communications, special issue on ``Recent Advances on 5G Communications'' and a Guest Editor of  IEEE Access, special issue on ``Modelling, Analysis, and Design of 5G Ultra-Dense Networks'', in 2017. He has been a member of Technical Program Committees for several IEEE conferences such as ICC, GLOBECOM, WCNC, and VTC.
\end{IEEEbiography}

\begin{IEEEbiography}
{Symeon Chatzinotas} is currently Full Professor / Chief Scientist I and Head of the SIGCOM Research Group at SnT, University of Luxembourg. He is coordinating the research activities on communications and networking, acting as a PI for more than 20 projects and main representative for 3GPP, ETSI, DVB.

In the past, he has been a Visiting Professor at the University of Parma, Italy, lecturing on “5G Wireless Networks”. He was involved in numerous R\&D projects for NCSR Demokritos, CERTH Hellas and CCSR, University of Surrey.

He was the co-recipient of the 2014 IEEE Distinguished Contributions to Satellite Communications Award and Best Paper Awards at EURASIP JWCN, CROWNCOM, ICSSC. He has (co-)authored more than 450 technical papers in refereed international journals, conferences and scientific books.

He is currently in the editorial board of the IEEE Transactions on Communications, IEEE Open Journal of Vehicular Technology and the International Journal of Satellite Communications and Networking.
\end{IEEEbiography}

\begin{IEEEbiography}
{Bj\"orn Ottersten} (S'87–M'89–SM'99–F'04) received the M.S. degree in electrical engineering and applied
physics from Linköping University, Linköping, Sweden, in 1986, and the Ph.D. degree in electrical
engineering from Stanford University, Stanford, CA, USA, in 1990. He has held research positions with
the Department of Electrical Engineering, Linköping University, the Information Systems Laboratory,
Stanford University, the Katholieke Universiteit Leuven, Leuven, Belgium, and the University of
Luxembourg, Luxembourg. From 1996 to 1997, he was the Director of Research with ArrayComm, Inc., a
start-up in San Jose, CA, USA, based on his patented technology. In 1991, he was appointed Professor of
signal processing with the Royal Institute of Technology (KTH), Stockholm, Sweden. Dr. Ottersten has
been Head of the Department for Signals, Sensors, and Systems, KTH, and Dean of the School of
Electrical Engineering, KTH. He is currently the Director for the Interdisciplinary Centre for Security,
Reliability and Trust, University of Luxembourg. He is a recipient of the IEEE Signal Processing Society
Technical Achievement Award, the EURASIP Group Technical Achievement Award, and the European
Research Council advanced research grant twice. He has co-authored journal papers that received the
IEEE Signal Processing Society Best Paper Award in 1993, 2001, 2006, 2013, and 2019, and 8 IEEE
conference papers best paper awards. He has been a board member of IEEE Signal Processing Society,
the Swedish Research Council and currently serves of the boards of EURASIP and the Swedish
Foundation for Strategic Research. Dr. Ottersten has served as Editor in Chief of EURASIP Signal
Processing, and acted on the editorial boards of IEEE Transactions on Signal Processing, IEEE Signal
Processing Magazine, IEEE Open Journal for Signal Processing, EURASIP Journal of Advances in Signal
Processing and Foundations and Trends in Signal Processing. He is a fellow of EURASIP.
\end{IEEEbiography}
\begin{IEEEbiography}
{M\'erouane Debbah} received the M.Sc. and Ph.D. degrees from the Ecole Normale Supérieure Paris-Saclay, France. He was with Motorola Labs, Saclay, France, from 1999 to 2002, and also with the Vienna Research Center for Telecommunications, Vienna, Austria, until 2003. From 2003 to 2007, he was an Assistant Professor with the Mobile Communications Department, Institut Eurecom, Sophia Antipolis, France. In 2007, he was appointed Full Professor at CentraleSupelec, Gif-sur-Yvette, France. From 2007 to 2014, he was the Director of the Alcatel-Lucent Chair on Flexible Radio. From 2014 to 2021, he was Vice-President of the Huawei France Research Center. He was jointly the director of the Mathematical and Algorithmic Sciences Lab as well as the director of the Lagrange Mathematical and Computing Research Center. Since 2021, he is Chief Research Officer at the Technology Innovation Institute in Abu Dhabi. He leads jointly the AI and Telecommunication centers. He has managed 8 EU projects and more than 24 national and international projects. His research interests lie in fundamental mathematics, algorithms, statistics, information, and communication sciences research. He is an IEEE Fellow, a WWRF Fellow, a Eurasip Fellow, an Institut Louis Bachelier Fellow and a Membre émérite SEE. He was a recipient of the ERC Grant MORE (Advanced Mathematical Tools for Complex Network Engineering) from 2012 to 2017. He was a recipient of the Mario Boella Award in 2005, the IEEE Glavieux Prize Award in 2011, the Qualcomm Innovation Prize Award in 2012, the 2019 IEEE Radio Communications Committee Technical Recognition Award and the 2020 SEE Blondel Medal. He received more than 20 best paper awards, among which the 2007 IEEE GLOBECOM Best Paper Award, the Wi-Opt 2009 Best Paper Award, the 2010 Newcom++ Best Paper Award, the WUN CogCom Best Paper 2012 and 2013 Award, the 2014 WCNC Best Paper Award, the 2015 ICC Best Paper Award, the 2015 IEEE Communications Society Leonard G. Abraham Prize, the 2015 IEEE Communications Society Fred W. Ellersick Prize, the 2016 IEEE Communications Society Best Tutorial Paper Award, the 2016 European Wireless Best Paper Award, the 2017 Eurasip Best Paper Award, the 2018 IEEE Marconi Prize Paper Award, the 2019 IEEE Communications Society Young Author Best Paper Award, the 2021 Eurasip Best Paper Award, the 2021 IEEE Marconi Prize Paper Award as well as the Valuetools 2007, Valuetools 2008, CrownCom 2009, Valuetools 2012, SAM 2014, and 2017 IEEE Sweden VT-COM-IT Joint Chapter best student paper awards. He is an Associate Editor-in-Chief of the journal Random Matrix: Theory and Applications. He was an Associate Area Editor and Senior Area Editor of the IEEE TRANSACTIONS ON SIGNAL PROCESSING from 2011 to 2013 and from 2013 to 2014, respectively. From 2021 to 2022, he serves as an IEEE Signal Processing Society Distinguished Industry Speaker
\end{IEEEbiography}
\end{document}